\documentclass[12pt]{article}
\usepackage[english]{babel}
\usepackage{amsthm, amsmath, amssymb, amsbsy, amsfonts, latexsym, euscript}

\newtheorem{theorem}{Theorem}
\newtheorem{definition}[theorem]{Definition}
\newtheorem{lemma}[theorem]{Lemma}

\newtheorem{corollary}[theorem]{Corollary}

\begin{document}

\begin{center}
{\Large\bf Convergence to equilibrium distribution.\\~\\
 Dirac fields coupled to a particle }
\end{center}
\vspace{1cm}
 \begin{center}
{\large T.V. Dudnikova}\\~\\
{\it M.V. Keldysh Institute of Applied Mathematics RAS\\
 Moscow 125047, Russia}\\~\\
e-mail:~tdudnikov@mail.ru
\end{center}
 \vspace{1cm}

 \begin{abstract}
For a system consisting of several Dirac fields and a particle, 
we study the Cauchy problem with random initial data.
We assume that the initial measure has zero mean value,
 a finite mean charge density, a translation-invariant covariance
 and satisfies  a  mixing condition.
The main result is the long-time convergence of distributions
of the random solutions to a limit Gaussian measure.
\medskip

\medskip

 \end{abstract}

\section{Introduction}

The paper is devoted to the  problem of long-time convergence to an
 equilibrium distribution in infinite-dimensional systems.
The statement of the problem in a general framework
and the overview of the first results are presented in \cite{DSi, DSS}.
In particular, for an ideal gas with infinitely many particles
and one-dimensional hard rods, this problem was studied in \cite{DS, DBS}.
Also, the convergence to equilibrium was established
for  one-dimensional chains of harmonic oscillators by Boldrighini and others in \cite{BPT},
for the harmonic crystals in \cite{DKM1}
 and for one-dimensional chains of anharmonic oscillators coupled
to heat baths by Jak\v{s}i\'c, Pillet and others (see, e.g., \cite{JP, EPR}).
In the systems described
by hyperbolic partial differential equations, the convergence analysis
was started  by Ratanov~\cite{R2} for wave equations.
 Later, the similar results were obtained for Klein--Gordon~\cite{K3, DKM2}
  and Dirac equations~\cite{DKM-03},
for a scalar field coupled to a harmonic crystal~\cite{DK}
and for the Klein--Gordon field coupled to a particle~\cite{D-2010}.
\medskip

In this paper,
we consider a system consisting of
a particle with position $q=(q_1,q_2,q_3)\in\mathbb{R}^3$
and
$N$ Dirac fields $\psi_1(x),\ldots,\psi_N(x)$, $x\in\mathbb{R}^3$,
where all
$\psi_n(x)=(\psi_{n1}(x),\ldots,\psi_{n4}(x))$ take values in $\mathbb{C}^4$,
 $n\in \overline{N}:=\{1,\ldots,N\}$.
The coupled dynamics is governed by the following equations
\begin{align}
i\,\dot\psi_n(x,t)&=(-i\,\alpha\cdot\nabla+\beta m_n)\psi_n(x,t)-q(t)\cdot\nabla\rho_n(x),
\quad n\in \overline{N},
\quad x\in \mathbb{R}^3,\quad t\in\mathbb{R},\label{1}\\
\ddot q(t)&=-V q(t)+\sum\limits_{n=1}^N
\left\langle\psi_n(\cdot\,,t),\nabla\rho_n\right\rangle.\label{2}
\end{align}
Here $m_n>0$, $V$ is a  positive  symmetric matrix,
$\rho_n \in C^1(\mathbb{R}^3;\mathbb{C}^4)$,
$\alpha=(\alpha_1,\alpha_2,\alpha_3)$,
$\nabla=(\partial_1,\partial_2,\partial_3)$,
$\partial_j=\partial/(\partial x_j)$, $j=1,2,3$,
$\alpha_j$ and $\beta$ are $4\times4$ Dirac matrices,
 ``$\cdot$'' stands for the scalar product
in the Euclidean space $\mathbb{R}^3$.
Here and below, the brackets $\langle\cdot\,,\cdot\rangle$ mean the following.
Set
 $$
 {\cal R}\psi_n:=
\left(\Re\psi_{n1},\dots,\Re\psi_{n4},\Im\psi_{n1},\dots,\Im\psi_{n4}\right)
\quad \mbox{for }\,\,\, \psi_n= (\psi_{n1},\dots\psi_{n4})\in \mathbb{C}^4,
\quad n\in \overline{N},
$$
 and   denote by  ${\cal R}^j\psi_n$ the
 $j$th component of the vector ${\cal R}\psi_n$, $j=1,...,8$.
 Then, for
$\psi=(\psi_1,\ldots,\psi_N)$ and
$\chi=(\chi_1,\ldots,\chi_N)$, we write
\begin{equation}\label{10}
\langle\psi,\chi\rangle:=\sum\limits_{n=1}^N\langle\psi_n,\chi_n\rangle:=
 \sum\limits_{n=1}^N\left({\cal R}\psi_n,{\cal R}\chi_n\right)=
\sum\limits_{n=1}^N\sum\limits_{j=1}^8\left({\cal R}^j\psi_n,{\cal R}^j\chi_n\right).
\end{equation}
Here and below, the brackets $(\cdot\,,\cdot)$
 mean  the  inner product in the real  Hilbert spaces
 $L^2\equiv L^2(\mathbb{R}^3)$, or in $L^2\otimes \mathbb{R}^8$, or in
some their   extensions.
The standard representation for the Dirac matrices $\alpha_j$ and $\beta$
(in $2\times 2$ blocks) is
\begin{equation}\label{matD}
\alpha_j=\begin{pmatrix}
           0 &\sigma_j \\
           \sigma_j & 0 \\
         \end{pmatrix},\quad j=1,2,3;\qquad \beta=\begin{pmatrix}
                                                         \mathrm{I} & 0 \\
                                                            0 & -\mathrm{I} \\
                                                          \end{pmatrix},
\end{equation}
where $\mathrm{I}$ is the unit $2\times2$ matrix,
$\sigma_j$ are the Pauli matrices,
$$
\sigma_1=\begin{pmatrix}
           0 & 1 \\
           1 & 0 \\
         \end{pmatrix},\quad \sigma_2=\begin{pmatrix}
           0 & -i \\
           i & 0 \\
         \end{pmatrix},\quad \sigma_3=\begin{pmatrix}
           1 & 0 \\
           0 & -1 \\
         \end{pmatrix}.
$$
Below by $\mathrm{I}$ we denote the unit $n\times n$ matrix with arbitrary $n=2,3,\ldots$.
The matrices $\alpha_j$ and $\beta$ are Hermitian and satisfy the anticommutation relations,
\begin{equation}\label{rel}
\alpha^*_j=\alpha_j, \quad j=0,1,2,3, \quad \mbox{where }\,\,\alpha_0:=\beta,
\quad
\alpha_j\alpha_k+\alpha_k\alpha_j=2\delta_{jk},
\quad j,k=0,1,2,3.
\end{equation}

In Eqn~\eqref{1} instead of matrices $\alpha_j$ and $\beta$
of the standard representation~\eqref{matD}
we can take matrices depending on the number $n=1,\ldots,N$ and satisfying relations~\eqref{rel}.
Then, all results remain true.
For simplicity of exposition, we assume that $\alpha_j$ and $\beta$
do not depend on $n$ and have form~\eqref{matD}.

Now we show that the system~\eqref{1}--\eqref{2} has a Hamiltonian structure.
Indeed, we put
$
A_1:=\alpha_1\partial_1+\alpha_3\partial_3$, $A_{2n}:=-i\alpha_2\partial_2+\beta m_n.
$
Note that $\alpha_1,\alpha_3,i\alpha_2,\beta\in\mathbb{R}^4\times\mathbb{R}^4$
and $(i\alpha_2)^\mathrm{T}=-i\alpha_2$, where $\mathrm{T}$ denotes the transposition of the matrix.
Then,
$\forall \phi_1,\phi_2\in C_0^\infty(\mathbb{R}^3;\mathbb{R}^4)$,
$\langle\phi_1,A_1\phi_2\rangle=-\langle A_1\phi_1,\phi_2\rangle$
and $\langle\phi_1,A_{2n}\phi_2\rangle=\langle A_{2n}\phi_1,\phi_2\rangle$.
Denote $\phi_n:=\Re\psi_n$, $\pi_n:=\Im\psi_n$, $\mu_n:=\Re\rho_n$, $\nu_n:=\Im\rho_n$,
where $\Re\psi_n=(\Re\psi_{n1},\ldots,\Re\psi_{n4})$,
$\Im\psi_n=(\Im\psi_{n1},\ldots,\Im\psi_{n4})$.
 We define $\Re\rho_n$ and $\Im\rho_n$ by a similar way.
Then system~\eqref{1}, \eqref{2} becomes
\begin{equation}\label{2.1}
\left\{
\begin{aligned}
\dot\phi_n(x,t) &=-A_1\phi_n(x,t)
+A_{2n}\pi_n(x,t)-q(t)\cdot\nabla\nu_n(x),\quad x\in \mathbb{R}^3,\quad t\in\mathbb{R}, \\
 \dot q(t)&=p(t),\\
  \dot\pi_n(x,t) &=-A_{2n}\phi_n(x,t)
  -A_1\pi_n(x,t)+q(t)\cdot\nabla\mu_n(x),\quad n=1,\ldots,N,
  \\
  \dot p(t)&=-Vq(t)+\sum\limits_{n=1}^N\Big(\left(\phi_n(\cdot,t),\nabla\mu_n\right)
  +\left( \pi_n(\cdot,t),\nabla\nu_n\right)\Big).
\end{aligned}
\right.
\end{equation}
Write $\phi=(\phi_1,\ldots,\phi_N)$, $\pi=(\pi_1,\ldots,\pi_N)$.
Then,  system~\eqref{2.1} can be represented as the Hamiltonian system with the Hamiltonian functional
\begin{align}\label{H}
\mathrm{H}(\phi,q,\pi,p)&=\sum\limits_{n=1}^N\frac12
\Big(\left(\phi_n, A_{2n}\phi_n\right)
+\left(\pi_n, A_{2n}\pi_n\right)
+2\left(\phi_n, A_1\pi_n\right)\Big)\notag\\
&+\frac12\left(q\cdot Vq+|p|^2\right)-\sum\limits_{n=1}^Nq\cdot\Big(
\left(\phi_n,\nabla \mu_n\right)+\left(\pi_n,\nabla\nu_n\right)\Big),
\end{align}
because the right hand side of equations in  \eqref{2.1} is equal to
$\delta \mathrm{H}/(\delta\pi_n)$, $\partial \mathrm{H}/(\partial p)$,
$-\delta \mathrm{H}/(\delta\phi_n)$, $-\partial \mathrm{H}/(\partial q)$, respectively.
Another words, the system~\eqref{2.1} has a form
\begin{equation}\label{2.4}
\dot Y(t)=J\,\mathcal{D}\,\mathrm{H}(Y(t)),\quad J:=\begin{pmatrix}
                                           0 & \mathrm{I}  \\
                                           -\mathrm{I} & 0  \\
                                         \end{pmatrix},
                                         \qquad Y=(\phi,q,\pi,p),
\end{equation}
where $\mathcal{D}\mathrm{H}$ is the Fr\'{e}chet derivative with respect to
$\phi,q,\pi, p$ of the Hamiltonian  defined in \eqref{H},
$\mathrm{I}$ is the unit $(4N+3)\times (4N+3)$ matrix.
\smallskip

\textbf{Remark}. Formula~\eqref{2.4} implies the energy conservation law,
$\mathrm{H}(Y(t))=\mathrm{H}(Y(0))$, $t\in\mathbb{R}$.
Indeed,
$$
\frac{d}{dt}\,\mathrm{H}(Y(t))=\left(\mathcal{D}\mathrm{H}(Y(t)),\dot Y(t)\right)
=\Big(\mathcal{D}\mathrm{H}(Y(t)),J\,\mathcal{D}\mathrm{H}(Y(t))\Big)=0,\quad t\in\mathbb{R},
$$
since the operator $J$ is skew-symmetric and
$\mathcal{D}\mathrm{H}(Y(t))\in\mathcal{E}$ for $Y(t)\in \mathcal{E}$.
\smallskip

Below we use notation $Y=(\psi,q,p)$, where $\psi=(\psi_1,\ldots,\psi_N)$
takes the values in $\mathbb{C}^{4N}$. 

\subsection{Conditions on the system. Cauchy problem}
We impose the  conditions~\textbf{A1}--\textbf{A3} on  the
coupling function $\rho(x)=(\rho_1(x),\ldots,\rho_N(x))$,
$\rho_n(x)=(\rho_{n1}(x),\ldots,\rho_{n4}(x))\in \mathbb{C}^4$,
$n=1,\ldots,N$, $x\in\mathbb{R}^3$, and on the
matrix $V$.
\begin{itemize}
\item[{\textbf{A1}.}]
$\rho(-x)=\rho(x)$, $\rho\in C^1(\mathbb{R}^3;\mathbb{C}^4)$,
$\rho(x)=0$ for $|x|\ge R_\rho$.

\item[{\textbf{A2}.}]
The matrix $V-m_*^2\mathrm{I}-K$ is positive definite,
where $m_*=\min\{m_n,n=1,\ldots,N\}$ and the matrix $K=(K_{ij})_{i,j=1}^3$
has entries
$$
K_{ij}:=\sum\limits_{n=1}^N\frac{m_n}{(2\pi)^{3}}
\int_{\mathbb{R}^3}\frac{k_ik_j\,\mathcal{B}_n(k)}{k^2+m_n^2-m_*^2}\,dk,
\quad i,j=1,2,3.
$$
Here $\mathcal{B}_n(k):=\widehat{\rho}_n(k)\cdot\beta\widehat{\rho}_n(k)$,
where $\widehat{\rho}_n(k)=\int e^{ik\cdot x}\rho_n(x)\,dx$
is the Fourier transform of the function $\rho_n$.
By \eqref{matD},
$\mathcal{B}_n(k)=|\widehat{\rho}_{n1}(k)|^2+|\widehat{\rho}_{n2}(k)|^2
-|\widehat{\rho}_{n3}(k)|^2-|\widehat{\rho}_{n4}(k)|^2$.

\item[{\textbf{A3}.}]
$\mathcal{B}_n(k)>0$  for all   $k\in\mathbb{R}^3\setminus\{0\}$ and $n=1,\ldots,N$.
\end{itemize}

We denote $\psi(x,t)=(\psi_1(x,t),\ldots,\psi_N(x,t))$,
$\psi^0(x)=\left(\psi^0_1(x),\ldots,\psi_N^0(x)\right)$.
We study the Cauchy problem for the system \eqref{1}--\eqref{2}
with initial data
\begin{equation}\label{ID}
\psi(x,0)=\psi^0(x),\quad
q(0)=q^0,\quad \dot q(0)=p^0.
\end{equation}
We denote
$Y_0\equiv(\psi^0(x),q^0,p^0)$, $Y(t)\equiv(\psi(x,t),q(t),\dot q(t))$.
Then  the system \eqref{1}--\eqref{2} writes as
\begin{equation}\label{CP}
\dot Y(t)={\cal F}(Y(t)),\quad t\in\mathbb{R};\quad Y(0)=Y_0.
\end{equation}

We assume that the initial date $Y_0$  belongs to the phase space ${\cal E}$.
\begin{definition}\label{d1.1}
Denote by $\mathcal{H}\equiv [L^2_{\mathrm{loc}}(\mathbb{R}^3;\mathbb{C}^{4})]^N$
the Fr\'echet space of complex- and vector-valued functions $\psi(x)=(\psi_1(x),\ldots,\psi_N(x))$,
endowed with local (charge) seminorms
$$
\Vert \psi\Vert^2_{0,R}\equiv
\int\limits_{|x| <R} |\psi(x)|^2\,dx   <\infty,   \quad \forall R>0.
$$
The phase space  ${\cal E} \equiv \mathcal{H}\oplus \mathbb{R}^3\oplus \mathbb{R}^3$
is the Fr\'echet space of vectors $Y\equiv(\psi(x),q,p)$,
endowed with the local seminorms
$
\Vert Y\Vert^2_{{\cal E},R}=
\Vert\psi\Vert^2_{0,R}+|q|^2+|p|^2$, $\forall R>0$.
\end{definition}
 \begin{lemma}    \label{p1.1}
Let conditions {\bf A1}--{\bf A3} hold. Then (i)
for every $Y_0 \in {\cal E}$, the Cauchy problem~\eqref{CP}
has a unique solution $Y(t)\in C(\mathbb{R}, {\cal E})$.
(ii) For every $t\in \mathbb{R}$, the operator $U(t):Y_0\mapsto  Y(t)$
 is continuous on ${\cal E}$.  Moreover, for every $R>R_\rho$ and $T>0$,
$$
\sup\limits_{|t|\le T}
\Vert U(t) Y_0\Vert_{{\cal E},R}\le
C(T)\Vert Y_0\Vert_{{\cal E},R+T}.
$$
\end{lemma}

Lemma~\ref{p1.1} follows from \cite[Thms. V.3.1, V.3.2]{Mikh})
because the speed of propagation for Eqs.~\eqref{1}--\eqref{2} is finite
by condition~\textbf{A1} and the Duhamel integral representation~\eqref{6.4} (see also Lemma~\ref{p1.0} below).

Let us choose a function
$\zeta(x)\in C_0^\infty(\mathbb{R}^3)$ such that $\zeta(0)\ne 0$.
Denote by $H^s_{\rm loc}(\mathbb{R}^3)$, $s\in \mathbb{R}$,  the local Sobolev spaces,
i.e., the Fr\'echet spaces
of distributions $\psi\in D'(\mathbb{R}^3)$ with the finite seminorms
$
\Vert \psi\Vert _{s,R}:= \Vert\Lambda^s\left(\zeta(x/R)\psi\right)\Vert_{L^2(\mathbb{R}^3)},
$
where $\Lambda^s \psi:=F^{-1}_{k\to x}(\langle k\rangle^s\widehat{\psi}(k))$,
$\langle k\rangle:=\sqrt{|k|^2+1}$, and $\widehat{\psi}:=F \psi$ is the Fourier transform
of a tempered distribution $\psi$. For $\psi\in C_0^\infty(\mathbb{R}^3)$,
we write $F\psi(k)= \int e^{i k\cdot x} \psi(x) dx$.
\begin{definition}
Write $\mathcal{H}^s:=\left[H^s_{\mathrm{loc}}(\mathbb{R}^3;\mathbb{C}^4)\right]^N$.
 We denote
${\cal E}^{s}:= \mathcal{H}^s\oplus \mathbb{R}^3\oplus \mathbb{R}^3$, $s\in \mathbb{R}$.
\end{definition}

Using standard techniques of pseudodifferential operators
and Sobolev's Theorem (see, e.g., \cite{H3}), it is possible to prove that
 ${\cal E}^0={\cal E}\subset {\cal E}^{-\varepsilon}$ for every $\varepsilon>0$,
and the embedding  is compact.

\subsection{Random solution. Convergence to equilibrium}

Let $(\Omega,\Sigma,P)$ be a probability space  with expectation $\mathbb{E}$
and ${\cal B}({\cal E})$ denote the Borel $\sigma$-algebra in ${\cal E}$.
We assume that $Y_0=Y_0(\omega,x)$ in the problem~\eqref{CP}
is a measurable random function
with values in $({\cal E},\,{\cal B}({\cal E}))$.
In other words, $(\omega,x)\mapsto Y_0(\omega,x)$ is a measurable map
$\Omega\times \mathbb{R}^3\to \mathbb{C}^{4N}\oplus\mathbb{R}^6$
 with respect to the (completed) $\sigma$-algebra
$\Sigma\times{\cal B}(\mathbb{R}^3)$ and ${\cal B}(\mathbb{C}^{4N}\oplus\mathbb{R}^6)$.
Then $Y(t)=U(t) Y_0$ is also a measurable random function with values in
$({\cal E},{\cal B}({\cal E}))$ by Lemma~\ref{p1.1}.
We denote by $\mu_0(dY_0)$ a Borel probability measure in ${\cal E}$
giving the distribution of  $Y_0$.
Without loss of generality, we assume
$(\Omega,\Sigma,P)=({\cal E},{\cal B}({\cal E}),\mu_0)$
and $Y_0(\omega,x)=\omega(x)$ for
$\mu_0(d\omega)\times dx$-almost all $(\omega,x)\in{\cal E}\times \mathbb{R}^3$.

We identify the complex and real spaces $\mathbb{C}^4\equiv\mathbb{R}^8$,
and  $\otimes$ stands for the tensor product of real vectors.
\begin{definition}
 $\mu_t$ is a Borel probability measure in ${\cal E}$
which gives the distribution of $Y(t)$:
$\mu_t(B) = \mu_0(U(-t)B)$,
$\forall B\in {\cal B}({\cal E})$,  $t\in  \mathbb{R}$.
 The correlation functions of measure $\mu_t$ are
defined by
$
Q_t(x,y)\equiv \mathbb{E} \left(Y(x,t)\otimes Y(y,t)\right)\quad
\mbox{for almost all }\,\,\, x,y\in \mathbb{R}^3.
$
\end{definition}

Our main objective is to prove
 the weak convergence of the measures $\mu_t$
in the Fr\'echet spaces ${\cal E}^{-\varepsilon }$ for each  $\varepsilon>0$,
\begin{equation}\label{1.8}
\mu_t\rightharpoondown \mu_\infty  \qquad \mbox{as }\quad t\to \infty,
\end{equation}
where $\mu_\infty$ is a limit measure on ${\cal E}$.
By definition, this means the convergence of the following integrals
$$
 \int f(Y)\mu_t(dY)\rightarrow
 \int f(Y)\mu_\infty(dY)\quad \mbox{as}\quad t\to \infty,
$$
 for any bounded continuous functional $f(Y)$
 on  ${\cal E}^{-\varepsilon }$.
 Moreover, in Section~4 we prove the convergence of the correlation
functions of the measures $\mu_t$ to a limit as $t\to\infty$.
\smallskip

Using the methods of \cite{VF} (Russian ed., Appendix~II,
 and English ed., Theorem XII.5.2) and
 the technique  of \cite{DKM2, DKM-03}, we conclude that
the convergence~\eqref{1.8} follows from the following three assertions.
\begin{itemize}
\item[\bf I.] The family of measures $\mu_t$, $t\geq 0$, is weakly
compact in ${\cal E}^{-\varepsilon }$ for each  $\varepsilon>0$.

\item[\bf  II.] The correlation functions of $\mu_t$ converge to a limit,
\begin{equation*}
Q_t(x,y)\equiv \int\left( Y(x) \otimes Y(y)\right)\mu_t(Y)
\to Q_\infty(x,y),\quad t\to\infty.
\end{equation*}
\item[\bf  III.]
The characteristic functionals of $\mu_t$ converge to a Gaussian functional,
\begin{equation}\label{2.6i}
 \widehat{\mu}_t(Z) = \int \exp\left({i\langle Y,Z\rangle}\right)\mu_t(dY)
\rightarrow  \exp\left\{-\frac{1}{2}{\cal Q}_\infty (Z,Z)\right\},
\quad t\to\infty,
\end{equation}
where ${\cal Q}_\infty$ is the quadratic form with the  integral kernel
$Q_\infty(x,y)$.
\end{itemize}
\smallskip

Let us explain the main idea of the proof.
At first, we derive the decay of the order $(1+|t|)^{-3/2}$
for the local charge of the solution $Y(t)$ to problem~\eqref{CP} assuming
that the initial date $Y_0$ has a compact support
(see Theorem~\ref{l5.1}).
Then, we apply the integral representation~\eqref{6.4}
of $Y(t)$ and prove the uniform bound~\eqref{7.1.1}
for the mean local charge density with respect to the measure
 $\mu_t$, $t\ge0$.
Finally, property {\bf I} follows from the Prokhorov compactness
theorem; see \cite[Lemma II.3.1]{VF}.

To prove the assertions \textbf{II} and \textbf{III},
 we derive the asymptotic behavior of the solution
$Y(t)$ (see Corollary~\ref{c7.2}) of the form
\begin{equation}\label{0.1}
\langle Y(t),Z\rangle\sim
\sum\limits_{n=1}^N\left\langle W_n(t)\psi_n^0,\chi_n^Z\right\rangle,
\quad t\to\infty,
\end{equation}
where $W_n(t)$ is a solving operator to the Cauchy problem
for the free Dirac equation~\eqref{3},
the functions $\chi_n^Z$ are expressed by $Z\in C_0^\infty(\mathbb{R}^3)\oplus \mathbb{R}^6$ (see formula~\eqref{hn}).
Finally, we apply the results of \cite{D-2010}, where the weak convergence
of the statistical solutions
is proved for free Dirac equations.

\section{Main result}
Recall that the initial data $Y_0$ in the problem~\eqref{CP}
is a random function with a distribution $\mu_0$.
We write $\nu_0:=P\mu_0$, where $P:(\psi^0,q^0,p^0)\in {\cal E}
\to\psi^0\in \mathcal{H}$.

\begin{definition}
The correlation functions of the measure $\nu_0$ are
defined by the rule
$$
  Q_{0,nn'}^{\nu,ij}(x,y)\equiv
\int {\cal R}^i\psi^0_n(x){\cal R}^j\psi^0_{n'}(y)\,\nu_0(d\psi^0)
\quad\mbox{for almost all }\,\, x,y\in \mathbb{R}^3,\quad i,j=1,\ldots,8,
$$
$n,n'\in \overline{N}$,
provided that the expectations in the right hand side are finite.
\end{definition}

\subsection{Conditions on the initial measure}
We assume that  the initial measure $\mu_0$ satisfies
conditions {\bf S1}--{\bf S5}.
\begin{itemize}
\item[\bf S1.] $\mu_0$ has zero expectation value,
$\mathbb{E} \left(Y_0(x)\right)\equiv \int Y_0(x)\,\mu_0(dY_0)= 0$, $x\in \mathbb{R}^3$.

\item[\bf S2.]
$\mu_0$ has finite mean charge density, i.e.,
\begin{equation}\label{med}
\mathbb{E}\left(|\psi^0(x)|^2\right)\le e_0 <\infty,\quad \mathbb{E}(|q^0|^2+|p^0|^2)<\infty.
\end{equation}

\item[\bf S3.] The correlation functions of the measure $\nu_0$
are translation invariant, i.e.,
\begin{equation*}
Q_{0,nn'}^{\nu,ij}(x,y)=
q^{\nu,ij}_{0,nn'}(x-y),\quad x,y\in \mathbb{R}^3,\quad n,n'\in \overline{N},
\quad i,j=1,\ldots,8.
\end{equation*}
\item[\bf S4.] The correlation functions of $\nu_0$ obey the bound
\begin{equation*}
|q_{0,nn'}^{\nu,ij}(x)|  \le h(|x|),\quad x\in\mathbb{R}^3,\quad n,n'\in \overline{N},
\quad i,j=1,\ldots,8,
\end{equation*}
where $h$ is a nonnegative bounded function and $r^2 h(r)\in L^1(0,+\infty)$.
\end{itemize}
\begin{lemma} \label{l4.1}
Let conditions  {\bf S1}--{\bf S4} hold.
Then $\widehat{q}_{0,nn'}^{\nu,ij}\in L^1(\mathbb{R}^3)$ for any $i,\,j,\,n,\,n'$.
\end{lemma}
\begin{proof}
Conditions \textbf{S1}--\textbf{S4} imply
$$
\int\limits_{\mathbb{R}^3} |q_{0,nn'}^{\nu,ij}(x)|^p\,dx\le
C \int\limits_{\mathbb{R}^3}h^p(|x|)\,dx\le  C_1
 \int\limits_0^{+\infty} r^2 h(r)\,dr <\infty,\quad p\ge1.
$$
Hence,
\begin{equation}\label{4.7}
 q_{0,nn'}^{\nu,ij}\in L^p(\mathbb{R}^3),\quad p\ge1.
\end{equation}
Denote $\widehat{q}_{0,nn'}^{\nu,ij}(k)
=F_{x\to k}\left[{q}_{0,nn'}^{\nu,ij}(x)\right]$. We note that, due to condition \textbf{S3},
\begin{equation}\label{15}
\int\widehat{\mathcal{R}^i\psi_n}(k)\widehat{\mathcal{R}^j\psi_{n}}(k')
\,\nu_0(d\psi)
=F_{x\to k,\, y\to k'}\left[Q_{0,nn'}^{\nu,ij}(x,y)\right]
=(2\pi)^3\delta(k+k')\widehat{q}_{0,nn'}^{\nu,ij}(k).
\end{equation}
On the other hand, by Bohner's theorem, $\widehat{q}_{0}^{\nu}dk=(\widehat{q}_{0,nn'}^{\nu,ij}(k))dk$
is a nonnegative matrix-valued measure on $\mathbb{R}^3$,
since condition~\textbf{S3} implies that for any
$\chi=(\chi_1,\ldots,\chi_N)\in \mathcal{D}_0:=\left[C_0^\infty(\mathbb{R}^3;\mathbb{C}^4)\right]^N$,
\begin{align*}
\int\left|\langle \psi,\chi\rangle\right|^2\nu_0(d\psi)
=(2\pi)^{-3}\sum\limits_{n,n'=1}^N\sum\limits_{i,j=1}^8
\int_{\mathbb{R}^3}
\overline{\widehat{\mathcal{R}^i\chi_n}}(k)\,
\widehat{q}_{0,nn'}^{\nu,ij}
(k)\,\widehat{\mathcal{R}^j\chi_{n'}}(k)\,dk\ge0.
\end{align*}
In turn, condition {\bf S2} implies that the total measure
 $\widehat{q}_0^\nu(\mathbb{R}^3)$ is finite.
On the other hand, relation~\eqref{4.7} for $p=2$
gives $\widehat{q}^{\nu,ij}_{0,nn'}\in  L^2(\mathbb{R}^3)$.
Hence, $\widehat{q}^{\nu,ij}_{0,nn'}\in  L^1(\mathbb{R}^3)$.
\end{proof}
\begin{corollary}
The bound~\eqref{4.7}  with $p=1$ and the Hausdorff--Young inequality imply
\begin{equation}\label{16}
\left|\left\langle
Q_{0,nn'}^{\nu,ij}(x,y),f(x)g(y)\right\rangle\right|
\le C \Vert f\Vert_{L^2}\, \Vert g\Vert_{L^2}\quad\forall f,g\in L^2.
\end{equation}
\end{corollary}

To prove the convergence of correlation functions and the compactness of
 $\{\mu_t,\,t\ge0\}$,
we impose conditions~{\bf S1}--{\bf S4}.
If the  initial measure is Gaussian, then for the proof of assertion~\eqref{1.8}
also it suffices to impose {\bf S1}--{\bf S4}.
However, to prove \eqref{1.8} in the case of non-Gaussian initial measures,
we need a stronger condition ${\bf S5}$
than ${\bf S4}$. To state it, we define a mixing condition for the measure $\nu_0$.
\begin{definition}\label{dmix}
We denote by $\sigma({\cal A})$
the $\sigma$-algebra  in ${\cal H}$ generated by the
linear functionals $\psi\mapsto\, \langle \psi,\chi\rangle$,
where  $\chi\in {\cal D}_0$ with
$\mathrm{supp}\chi\subset {\cal A}\subset \mathbb{R}^3$.
We define the $\alpha$- and $\varphi$-mixing coefficients
of a probability measure $\nu_0$ on ${\cal H}$
by the rule (cf \cite[Def. 17.2.2]{IL} and \cite{Ros})
\begin{equation*}
\alpha(r):=
|\nu_0(A\cap B) - \nu_0(A)\nu_0(B)|,\quad
\varphi(r):=
\frac{| \nu_0(A\cap B) - \nu_0(A)\nu_0(B)|}{ \nu_0(B)},\quad r\ge0,
\end{equation*}
where the supremum is taken over all sets
$A\in\sigma({\cal A})$ and $B\in\sigma({\cal B})$  and
all pairs of open convex subsets
 ${\cal A}, {\cal B} \subset \mathbb{R}^3$ at a distance
$d({\cal A},{\cal B})\geq r$.
 We say that the measure $\nu_0$ satisfies
the $\alpha$-mixing ($\varphi$-mixing) condition if
$\alpha(r)\to 0$ ($\varphi(r)\to 0$) as $r\to\infty$.
\end{definition}
{\bf S5.} The measure $\nu_0$  satisfies the $\varphi$-mixing condition, and
\begin{equation}\label{1.12}
 r^2\varphi^{1/2}(r)\in L^1[0,+\infty).
\end{equation}

{\bf Remarks}. (1) Instead of the $\varphi$-mixing condition,
it suffices to assume the $\alpha$-mixing condition~\cite{Ros}
together with a higher degree ($>2$) in the bound~\eqref{med}, i.e.,
to assume that there is  a $\delta$, $\delta >0$, such that
\begin{equation}\label{med'}
\mathbb{E}\left(|\psi^0(x)|^{2+\delta}\right)\le e_\delta <\infty.
\end{equation}
In this case,  we assume that
$r^2\alpha^{p}(r)\in L^1[0,+\infty)$ with $p=\min\{\delta/(2+\delta),1/2\}$
(cf. bound~\eqref{1.12}).
Furthermore, the $\alpha$- and $\varphi$-mixing conditions can be weakened
by a similar way as in \cite{D-Mixing-23}.

(2) Conditions \textbf{S2} and \textbf{S5} implies \textbf{S4}, where
 $h(r)=Ce_0\varphi^{1/2}(r)$ (in the case of the $\varphi$-mixing)
 with the constant $e_0$ from bound~\eqref{med}, or $h(r)=Ce_\delta^{2/(2+\delta)}\alpha^{2/(2+\delta)}(r)$
(in the case of the $\alpha$-mixing)
with the constant $e_\delta$ from \eqref{med'}.
This follows from \cite[Theorems 17.2.2, 17.2.3]{IL}.
\smallskip

Before to state the main result, we formulate the result of \cite{DKM-03}
on the statistical stabilization for the free Dirac fields.

\subsection{Convergence to equilibrium for the free Dirac equation}
Let us fix a number $n\in \overline{N}$ and write $l_n(\nabla):=\alpha\cdot\nabla+i\,\beta m_n$.
The  Cauchy problem for the free Dirac equation
 reads
\begin{equation}\label{3}
\left(\partial_t+l_n(\nabla)\right)\psi_n(x,t)=0,\quad t>0,
\quad \psi_n(x,0)=\psi^0_n(x), \quad x\in \mathbb{R}^3.
\end{equation}
In the Fourier transform, the solution to problem~\eqref{3} is
$\widehat{\psi}_n(k,t)=e^{i(\alpha\cdot k-\beta m_n)t}\,\widehat{\psi}^0_n(k)$. Hence,
\begin{equation}\label{2.8}
\Vert W_n(t) \psi^0_n\Vert=\Vert\psi^0_n\Vert,\quad\psi_n^0\in L^2,\quad t\in\mathbb{R},
\end{equation}
by  relations~\eqref{rel}. Here and below, $\Vert\cdot\Vert$ denotes the norm in $L^2$.
Below, we apply the following
well-known bounds (see, e.g., \cite{RS3}). Let $\psi^0_n=0$ for $|x|\ge R_1$. Then
for any $R>0$,
\begin{equation}\label{eqfreeD}
\Vert W_n(t)\psi^0_n\Vert_{0,R}
\le C (1+t)^{-3/2}\Vert\psi_n^0\Vert_{0,R_1},\quad t\ge0.
\end{equation}
Formulas  \eqref{matD} and \eqref{rel} imply
$\left(\partial_t+l_n(\nabla)\right)
\left(\partial_t- l_n(\nabla)\right)=(\partial^2_t -\triangle+m^2_n)\mathrm{I}.
$
Then,
the fundamental solution ${\cal E}_n(x,t)$ of the Dirac operator,
 i.e., a solution of the  equation
$$
\left(\partial_t+l_n(\nabla)\right)
{\cal E}_n(x,t)=\delta(x,t)\mathrm{I},\quad {\cal E}_n(x,t)=0\quad
\mbox{for }\,\,\, t<0,
$$
has the form
${\cal E}_n(x,t)=\left(\partial_t -l_n(\nabla)\right)g_{t,n}(x)$, $t>0$,
where  $g_{t,n}(x)$
is a fundamental solution for the Klein--Gordon operator
$\partial^2_t -\triangle+m^2_n$, and $g_{t,n}$ vanishes for $t<0$.
Hence,  $W_n(t)$ is a convolution operator of the form
\begin{equation}\label{25}
W_n(t) \psi^0_n={\cal E}_n(\cdot,t)*\psi^0_n=
\left(\partial_t - l_n(\nabla)\right)g_{t,n}*\psi^0_n.
\end{equation}

\textbf{Remark}.
The function $g_{t,n}(x)$ is given by
$g_{t,n}(x)=F^{-1}_{k\to x}\left[\frac{\sin \omega_n(k) t}{\omega_n(k)}\right]$,
$\omega_n(k)\equiv \sqrt{|k|^2+m_n^2}$.
Then, by the Paley--Wiener Theorem (see, e.g., \cite[Theorem~7.3.1]{RS2}),
the function  $g_{t,n}(\cdot)$ is supported by the ball
$|x|\le t$.
The following lemma is proved in \cite{DKM-03}.
 \begin{lemma}    \label{p1.0}
 For any $\psi_0 \in \mathcal{H}_1:=L^2_{\mathrm{loc}}(\mathbb{R}^3;\mathbb{C}^4)$,
 there exists  a unique solution $\psi_n(\cdot,t)\in C(\mathbb{R},\,\mathcal{H}_1)$
 to the Cauchy problem~\eqref{3}.
 For any  $t\in \mathbb{R}$, the operator $W_n(t):\psi_n^0\mapsto  \psi_n(\cdot,t)$
 is continuous in $\mathcal{H}_1$ and for any $\psi_n^0\in \mathcal{H}_1$ and $R>R_\rho>0$,
 \begin{equation}\label{24}
 \Vert W_n(t)\psi_n^0\Vert_{0,R}\le \Vert\psi_n^0\Vert_{0,R+|t|},\quad t\in\mathbb{R}.
 \end{equation}
 \end{lemma}

Relation~\eqref{25} implies the following formula
\begin{equation}\label{ReW}
  \mathcal{R}(W_n(t)\psi^0_n)
=(\partial_t-\Lambda_n(\nabla))g_{t,n}*\mathcal{R}\psi_n^0,
 \quad \Lambda_n(\nabla):=\begin{pmatrix}
                            A_1 & -A_{2n}\\
                            A_{2n} & A_1 \\
                          \end{pmatrix},
\end{equation}
where $A_1\equiv A_1(\partial_1,\partial_3):=\alpha_1\partial_1+\alpha_3\partial_3$, $A_{2n}\equiv A_{2n}(\partial_2):=-i\alpha_2\partial_2+\beta m_n$.

\smallskip

We introduce the matrix-valued function

\begin{equation}\label{1.13}
Q^\nu_{\infty}(x,y)= \left(q^{\nu}_{\infty,nn'}(x-y)\right)_{n,n'=1}^N,
\quad
q^\nu_{\infty,nn'}(x):=\left\{
\begin{array}{ll}
q^\nu_{nn'}(x), & \mbox{if }\,m_n=m_{n'}, \\
 0, & \mbox{otherwise},
  \end{array}
   \right.
\end{equation}
 for almost all $x,y\in \mathbb{R}^3$,
 where  $q^\nu_{nn'}(x)=F^{-1}_{k\to x}[\widehat{q}^\nu_{nn'}(k)]$,
\begin{equation}\label{1.10}
 \widehat{q}^\nu_{nn'}(k):
=\frac{1}{2}\left(\widehat{q}^{\nu}_{0,nn'}(k)+
\widehat{\cal P}_n(k)\Lambda_n(-ik)\widehat{q}^\nu_{0,nn'}(k) \Lambda^\mathrm{T}_{n'}(ik)\right),
\quad n,n'\in \overline{N}.
\end{equation}
Here $\widehat{\cal P}_n(k)=1/(k^2+m_n^2)$,
and $\widehat{q}^\nu_{0,nn'}(k)=\left(\widehat{q}^{\nu,ij}_{0,nn'}(k)\right)_{i,j=1}^8$,  where $q^{\nu,ij}_{0,nn'}$ are the correlation functions
of the measure $\nu_0$.
Since $\Lambda^\mathrm{T}_n(ik)=-\Lambda_n(-ik)$, we have, formally,
\begin{equation*}
 q^\nu_{nn'} (x) =\frac{1}{2}\left(q^\nu_{0,nn'}(x)-
{\cal P}_n* \Lambda_n(\nabla) q^\nu_{0,nn'}(x)
\Lambda_{n'}(\stackrel{\leftarrow}\nabla)\right),
\end{equation*}
where ${\cal P}_n (z)=e^{-m_n|z|}/(4\pi|z|)$
is the fundamental solution for the
operator $-\Delta+m_n^2$, and
   $*$ stands for the convolution of distributions.


Denote by ${\cal Q}^\nu_{\infty} (\chi,\chi)$ a real quadratic form on
${\cal D}_0$ defined by
\begin{equation}\label{qpp}
{\cal Q}^\nu_{\infty} (\chi,\chi)=
\langle Q^\nu_{\infty}(x,y),\chi(x)\otimes \chi(y)\rangle
= \sum\limits_{n,n'=1}^N
\langle q^\nu_{\infty,nn'}(x-y),\chi_n(x)\otimes \chi_{n'}(y)\rangle,
\end{equation}
where $\langle\cdot\,,\cdot\rangle$ is defined in \eqref{10}.
\begin{lemma}\label{r2.9}
For all $i,j,n,n'$, the functions $\widehat{q}_{\infty,nn'}^{\nu,ij}$ are bounded.
Hence, the form ${\cal Q}^\nu_{\infty}$ is continuous on $L^2$.
\end{lemma}
\begin{proof}
Relation~\eqref{4.7} with $p=1$ implies that
$\widehat{q}^{\nu,ij}_{0,nn'}(k)$ are bounded.
Hence, \eqref{1.13} and \eqref{1.10} imply that $\widehat{q}_{\infty,nn'}^{\nu}$ are also bounded.
\end{proof}

We define the operator $\mathbf{W}(t)$
on the space ${\cal H}\equiv [\mathcal{H}_1]^N$ by the rule
\begin{equation}\label{W(t)}
\mathbf{W}(t)(\psi^0_1,\dots,\psi^0_N)
=(W_1(t)\psi^0_1,\dots,W_N(t)\psi^0_N).
\end{equation}
\begin{definition}
For a  probability  measure $\nu$ on  ${\cal H}$
we denote by $\widehat{\nu}$ the characteristic functional (the Fourier transform)
$$
\widehat{\nu}(\chi)\equiv\int\exp(i\langle \psi,\chi\rangle )\,\nu(d\psi),\quad
 \chi\in {\cal D}_0.
$$
A  measure $\nu$ is said to be Gaussian (with zero expectation) if
its characteristic functional has the form
$\widehat{\nu} (\chi)= \exp\left\{-\frac{1}{2}  {\cal Q}(\chi,\chi)\right\}$,
$\chi \in {\cal D}_0$,
where ${\cal Q}$ is a  real nonnegative quadratic form in ${\cal D}_0$.

A measure $\nu$ is called translation-invariant if
$\nu(T_h B)= \mu(B)$, $\forall B\in{\cal B}({\cal H})$, $h\in \mathbb{R}^3$,
where $T_h \psi(x)= \psi(x-h)$.
\end{definition}

The following result can be obtained by an easy adaptation
of the proof of  \cite[Theorem~A]{DKM-03},
where the result is proved in the case when $N=1$.
\begin{theorem} \label{l6.2}
Let conditions {\bf S1}--{\bf S3} and {\bf S5} hold. Then
 the measures $\nu_{t}\equiv \mathbf{W}(t)^*\nu_0$
weakly converge as $t\to\infty$ on the space ${\cal H}^{-\varepsilon}$
for each $\varepsilon>0$.
 The limit measure  $\nu_{\infty}$
is a translation-invariant Gaussian measure on ${\cal H}$.
 The characteristic functional of $\nu_{\infty}$ is of the form
$$
\widehat\nu_{\infty}(\chi)=\exp\left\{-\frac{1}{2}
{\cal Q}^\nu_{\infty}(\chi,\chi)\right\},\quad \chi\in {\cal D}_0,
$$
where ${\cal Q}^\nu_{\infty}(\chi,\chi)$ is defined in \eqref{qpp}.
\end{theorem}
\subsection{Statement of result}
To state the result~\eqref{1.8} precisely,
 we set ${\cal D}={\cal D}_0\oplus \mathbb{R}^3\oplus \mathbb{R}^3$
and $\langle Y,Z\rangle:=\langle \psi,\chi\rangle+q\cdot u+p\cdot v
$
for $Y=(\psi,q,p)\in {\cal E}$ and
$Z=(\chi,u,v)\in  {\cal D}$, i.e.,
$\chi=(\chi_1,\dots,\chi_N)\in {\cal D}_0$, $(u,v)\in \mathbb{R}^3\times \mathbb{R}^3$.
Denote
\begin{gather}
\chi^Z=(\chi^Z_1,\dots,\chi^Z_N),\quad \chi_n^Z:=\chi_n(x)+
\sum\limits_{r=1}^N\theta^{\chi_r}_{nr}(x)
+\Xi^0_n(x)\cdot  u+\Xi^1_n(x)\cdot v,
\quad n\in \overline{N}, \label{hn}\\
\label{teta}
\theta^\chi_{nr}(x):=\sum\limits_{k=1}^3\int\limits_0^{+\infty}
W_n(s)\,\Xi^0_{nk}(x)
\left\langle W_r(s)\,i\, \partial_k\rho_r,\chi\right\rangle\,ds,
\quad \chi\in C_0^\infty(\mathbb{R}^3;\mathbb{C}^4).
\end{gather}
Here
$\Xi^j_n(x)=\left(\Xi^j_{n1}(x),\Xi^j_{n2}(x),\Xi^j_{n3}(x)\right)$, $n\in \overline{N}$, $j=0,1$,
where
$\Xi^j_{nk}(x)$ are $\mathbb{C}^4$-valued functions of the form
\begin{equation}\label{xi}
\Xi^j_{nk}(x):=\sum\limits_{l=1}^3
\int\limits_0^{+\infty}N^{(j)}_{kl}(s)\,{W}_n(s)\,\partial_l \rho_n(x)\,ds,
 \quad x\in\mathbb{R}^3,\quad k=1,2,3,
\end{equation}
 with the matrix- and real-valued function  $N(t)=(N_{kl}(t))_{k,l=1}^3$ defined in Theorem~\ref{h-eq},
 $N^{(j)}_{kl}(t):=\frac{d^j}{dt^j} N_{kl}(t)$.
Denote by ${\cal Q}_{\infty} (Z,Z)$
a real quadratic form in ${\cal D}$ of the form
\begin{equation*}
{\cal Q}_{\infty} (Z,Z)=
{\cal Q}^\nu_{\infty} (\chi^Z,\chi^Z),
\end{equation*}
where $\chi^Z$ is defined in (\ref{hn}) and ${\cal Q}^\nu_{\infty}$
in \eqref{qpp}.
Our main result is the following theorem.
\begin{theorem}\label{tA}
Let conditions {\bf A1}--{\bf A3} be true. Then the following assertions hold.

(1) Let conditions {\bf S1}--{\bf S4} be fulfilled. Then
the correlation functions of $\mu_t$ converge to a limit, i.e.,
for any $Z_1,Z_2\in {\cal D}$,
\begin{equation}\label{concorf}
 \mathbb{E}\left(\langle Y(t),Z_1\rangle\langle Y(t),Z_2\rangle\right)
\to {\cal Q}_\infty(Z_1,Z_2),\quad t\to\infty.
\end{equation}

(2) Let conditions {\bf S1}--{\bf S3} and {\bf S5} be fulfilled.
Then the convergence in (\ref{1.8}) holds for any $\varepsilon>0$.
  The limit measure $ \mu_\infty $ is a Gaussian
measure on ${\cal E}$.
 The limit characteristic functional has the form
$$
\widehat{\mu}_\infty (Z)= \exp\left\{-\frac{1}{2}
{\cal Q}_\infty(Z,Z)\right\},\qquad Z \in  {\cal D}.
$$
 The measure $\mu_\infty$ is invariant, i.e.,
$U(t)^* \mu_\infty=\mu_\infty,\quad t\in \mathbb{R}$.
\end{theorem}

 The assertion (1) of Theorem \ref{tA} is proved in Section~4.2,
 the assertion~(2) can be derived
 from Lemmas~\ref{l2.1} and \ref{l2.2}.
\begin{lemma}\label{l2.1}
Let conditions {\bf S1}--{\bf S3} hold. Then the family of  measures $\{\mu_t, t\ge0\}$
 is weakly compact in the space
 ${\cal E}^{-\varepsilon}$ for any $\varepsilon>0$
 and
 \begin{equation}\label{7.1.1}
\sup\limits_{t\ge 0} \mathbb{E}\Vert U(t) Y_0\Vert^2_{{\cal E},R}
\le C(R)<\infty,\quad \forall R>0.
\end{equation}
\end{lemma}
\begin{lemma}\label{l2.2}
Let conditions {\bf S1}--{\bf S3} and {\bf S5} hold. Then
for any $Z\in {\cal D}$, the convergence~\eqref{2.6i} is true.
\end{lemma}
Lemma~\ref{l2.1} (Lemma~\ref{l2.2}) provides the existence (the uniqueness, resp.)
of the limit measure $\mu_\infty$.
They are proved in Sections~\ref{sec4.1} and \ref{sec4.3}, respectively.
Before proving Theorem~\ref{tA}, we study the long-time behaviour
of the random solutions to problem~\eqref{CP}.
\smallskip

{\bf Remark}.
All results remain true if we consider $N$ one-dimensional Dirac fields
$\psi_n(x)\in \mathbb{C}^{2}$, $n=1,\ldots, N$, $x\in\mathbb{R}$,
coupled to an  oscillator $q\in\mathbb{R}$,
$$
\left\{
\begin{aligned}
&i\,\dot\psi_n(x,t)=(-i\,\alpha\,\partial_x+\beta m_n)\psi_n(x,t)-q(t)\rho'_n(x),
\quad t\in\mathbb{R},\quad x\in \mathbb{R},
\quad n=1,\ldots,N,\\
&\ddot q(t)=-\kappa^2 q(t)+\sum\limits_{n=1}^N\langle\psi_n(\cdot,t),\rho'_n\rangle.
\end{aligned}
\right.
$$
Here $\kappa>0$, $\rho_n\in \mathbb{C}^2$, $\alpha=\begin{pmatrix}
               0 & 1 \\
               1 & 0 \\
             \end{pmatrix}$, $\beta=\begin{pmatrix}
            1 &0 \\
            0 & -1 \\
          \end{pmatrix}$,
 and $\kappa$ and $\rho_n$ satisfy restrictions similarly to conditions
          {\bf A1}--{\bf A3}.

\section{Asymptotic behavior of $Y(t)$ as $t\to\infty$}\label{sec3}

\begin{theorem}\label{l5.1}
Let conditions {\bf A1}--{\bf A3} hold and
let $Y_0\in \mathcal{E}$ be such that
\begin{equation*}
\psi^0(x)=0\quad \mbox{for }\,|x|>R_1,
\end{equation*}
with some $R_1>0$. Then
there exists a constant $C=C(R,R_1)>0$ such that
the following bound holds for every $R>0$,
\begin{equation}\label{delocen}
\Vert Y(t)\Vert_{{\cal E},R}
\le C\langle t\rangle^{-3/2}\Vert Y_0\Vert_{{\cal E},R_1},\quad t\ge0.
\end{equation}
\end{theorem}
\begin{proof}
To prove the bound~\eqref{delocen}, we follow the strategy of \cite{D-2025-ljm}.
Using the Duhamel representation for Eq.~\eqref{1} with initial data \eqref{ID}, we have
\begin{equation}\label{2.10}
\psi_n(x,t)={W}_n(t)\psi_n^0(x)+i\int\limits_0^t {W}_n(t-s)\nabla \rho_n(x)\cdot q(s)\,ds,
\end{equation}
where $W_n(t)$ is the solving operator to problem~\eqref{3}.
We substitute \eqref{2.10} in Eq.~\eqref{2} and obtain
\begin{align}
\ddot q(t)=-Vq(t)+\int\limits_0^t H(t-s) q(s)\,ds+F(t),\label{cs0}\\
F(t):=(F_1(t),F_2(t),F_3(t)),\quad
F_k(t):=\sum\limits_{n=1}^N
\left\langle \partial_k\rho_n,{W}_n(t)\psi_n^0\right\rangle,\label{41}\\
 H(t):=(H_{kl}(t))_{k,l=1}^3,\quad
H_{kl}(t):=\sum\limits_{n=1}^N
\left\langle \partial_k\rho_n,{W}_n(t)\,i\,\partial_l\rho_n\right\rangle.\label{2.14}
\end{align}
Denote by $S(t)=
\begin{pmatrix}
   \dot N(t) &  N(t) \\
   \ddot N(t) & \dot N(t) \\
 \end{pmatrix}
$ the solving operator to  problem~\eqref{cs0} with $F(t)\equiv0$
 and initial data
\begin{equation}\label{cs0-in}
q(t)|_{t=0}=q^0,\quad \dot q(t)|_{t=0}=p^0.
\end{equation}
Write $N^{(j)}(t):=\frac{d^j}{dt^j}N(t)$ for $j=0,1,2$.
To prove the decay~\eqref{delocen} for the solutions to problem~\eqref{cs0},
\eqref{cs0-in}, we apply the following bound (see Appendix~B).
 \begin{theorem}\label{h-eq}
 Let conditions {\bf A1}--{\bf A3} hold. Then
\begin{equation}\label{decayA}
|N^{(j)}(t)|\le C(1+t)^{-3/2},\quad j=0,1,2,\quad t\ge0.
\end{equation}
\end{theorem}

For the solutions to problem \eqref{cs0}, \eqref{cs0-in},
the following representation holds
\begin{equation}\label{3.4}
\begin{pmatrix}
  q(t) \\
  \dot{q}(t) \\
\end{pmatrix}
= S(t)\begin{pmatrix}
        q^0 \\
         p^0 \\
      \end{pmatrix}
+\int\limits_0^t S(\tau)\begin{pmatrix}
                  0 \\
                 F(t-\tau)\\
                 \end{pmatrix}
d\tau,\quad t>0.
\end{equation}
Due to bound \eqref{eqfreeD} and condition~\textbf{A1},
we have $|F(t)|\le C \langle t\rangle^{-3/2}\Vert \nabla\rho\Vert_{0,R_\rho}\Vert\psi^0\Vert_{0,R_1}$.
Hence, together with \eqref{3.4}, this implies bound~\eqref{delocen} for $q(t)$ and $\dot{q}(t)$.
For the field components $\psi_n(\cdot\,,t)$, the bound~\eqref{delocen}
follows from  representation~\eqref{2.10},
 bound~\eqref{eqfreeD} and  bound~\eqref{delocen} for $q(t)$:
\begin{align*}
&\Vert\psi_n(\cdot,t)\Vert_{0,R}
\le \Vert W_n(t)\psi_n^0\Vert_{0,R}
+\int\limits_0^t\Vert W_n(t-s)\nabla\rho_n\Vert_{0,R}\,|q(s)|ds
\\
&\le C\langle t\rangle^{-3/2}\Vert \psi_n^0\Vert_{0,R_1}
+C_1\int\limits_0^t\langle t-s\rangle^{-3/2}\langle s\rangle^{-3/2}\,ds\,
\Vert \nabla\rho_n\Vert_{0,R_\rho}
\Vert Y_0\Vert_{\mathcal{E},R_1}
\le C_2\langle t\rangle^{-3/2}\Vert Y_0\Vert_{\mathcal{E},R_1}.
\end{align*}
This completes the proof of bound~\eqref{delocen}.
\end{proof}

Now we specify the representation~\eqref{0.1} and prove it.
Write $q_k^{(j)}(t):=\frac{d^j}{dt^j}q_k(t)$ for $j=0,1$; $k=1,2,3$.
\begin{theorem}\label{l7.1}
Let conditions {\bf A1}--{\bf A3} and {\bf S1}--{\bf S4} be fulfilled. Then

(i) the following representation holds,
\begin{equation}\label{81}
q^{(j)}_k(t)=\sum\limits_{n=1}^N
\langle W_n(t)\psi_n^0,\Xi_{nk}^j\rangle+r_j(t),\quad \mbox{where }\,\,
\mathbb{E}\left|r_j(t)\right|^2
\le C(1+t)^{-1},
\end{equation}
$j=0,1$, $k=1,2,3$,  the functions $\Xi_{nk}^j$ are defined in \eqref{xi}.

(ii) Let $\chi\in C_0^\infty(\mathbb{R}^3; \mathbb{C}^4)$
with $\mathrm{supp}\, \chi\subset B_R:=\{x\in \mathbb{R}^3:|x|\le R\}$.  Then, for $n\in \overline{N}$
and $t\ge1$,
\begin{equation}\label{83}
\langle\psi_n(\cdot\,,t),\chi\rangle=
 \left\langle W_n(t)\psi^0_n,\chi\right\rangle+
\sum\limits_{r=1}^N
\left\langle W_r(t)\psi_r^0,\theta^\chi_{rn}\right\rangle+r(t),
\quad \mbox{where }\,\,
\mathbb{E}\left|r(t)\right|^2 \le C(1+t)^{-1},
\end{equation}
where the functions $\theta^\chi_{rn}(x)$, $r,n\in \overline{N}$,
are defined in \eqref{teta}.
\end{theorem}
\begin{proof} (i)
At first, relations~\eqref{3.4} and \eqref{41}  imply that
for each $k=1,2,3$, $j=0,1$,
\begin{equation}\label{84}
\mathbb{E} \Big|q^{(j)}_k(t)-\sum\limits_{n=1}^N\sum\limits_{l=1}^3\int\limits_0^t
\left\langle W_n(t-s)\psi^0_n, N^{(j)}_{kl}(s)\partial_l\rho_n
\right\rangle\,ds\Big|^2 =
\mathbb{E} \Big| N^{(j+1)}(t) q^0+N^{(j)}(t) p^0\Big|^2\le
C\langle t\rangle^{-3},
\end{equation}
by the bound~\eqref{decayA} and condition \textbf{S2}. Secondly,
\begin{equation*}
 \mathbb{E}\left|\int\limits_t^{+\infty}
\left\langle W_n(t-s)\psi_n^0,  N^{(j)}_{kl}(s)\partial_l\rho_n
\right\rangle\,ds\right|^2= \int\limits_t^{+\infty}  N^{(j)}_{kl}(s_1)\,ds_1
\int\limits_t^{+\infty} N^{(j)}_{kl}(s_2)A(t-s_1,t-s_2)\,ds_2,
\end{equation*}
where $A(t-s_1,t-s_2):=\mathbb{E}\left(
\left\langle W_n(t-s_1)\psi^0_n, \partial_l\rho_n \right\rangle
\left\langle W_n(t-s_2)\psi^0_n, \partial_l\rho_n \right\rangle
\right)$.
For any $t,s_1, s_2\in \mathbb{R}$, we have
\begin{equation*}
|A(t-s_1,t-s_2)|\le C
\sup_{\tau\in \mathbb{R}}\mathbb{E}|\langle W_n(\tau)\psi^0_n, \partial_l\rho_n\rangle|^2
\le C_1\sup_{\tau\in \mathbb{R}}\mathbb{E}\Vert W_n(\tau)\psi^0_n\Vert^2_{0,{R_\rho}}
\le C_2<\infty,
\end{equation*}
by the bound~\eqref{6.3}.
Hence, applying the bound~\eqref{decayA},  we obtain  that
\begin{equation}\label{48}
\mathbb{E}\left|\int\limits_t^{+\infty}
\langle W_n(t-s)\psi^0_n, N^{(j)}_{kl}(s)\partial_l\rho_n \rangle
\,ds\right|^2\le
C\left(\int\limits_t^{+\infty}\langle s\rangle ^{-3/2}\,ds\right)^2
\le C(1+t)^{-1}.
\end{equation}
Therefore, bounds \eqref{84} and \eqref{48}  imply
\begin{equation}\label{49}
 q^{(j)}_k(t)=\sum\limits_{n=1}^N\sum\limits_{l=1}^3\int\limits_0^{+\infty}
\left\langle W_n(t-s)\psi^0_n, N^{(j)}_{kl}(s)\partial_l\rho_n
\right\rangle ds +r_j(t),
\quad \mathbb{E}|r_j(t)|^2\le C\langle t\rangle^{-1}.
\end{equation}
To prove \eqref{81}, we introduce
 an operator $({W}_n(t))'$   adjoint to ${W}_n(t)$:
$$
\left\langle \psi, ({W}_n(t))'\phi\right\rangle
  =\left\langle {W}_n(t)\psi,\phi\right\rangle\quad
\mbox{for }\,\, \phi,\psi\in L^2.
$$
Note that for $\phi,\psi\in C_0^\infty(\mathbb{R}^3)$,
$$
\left\langle \psi, \frac{d}{dt}({W}_n(t))'\phi\right\rangle=
\left\langle \frac{d}{dt}{{W}}_n(t)\psi,\phi\right\rangle=
-\left\langle l_n(\nabla){W}_n(t)\psi,\phi\right\rangle
=\left\langle \psi, {l}_n(\nabla)({W}_n(t))'\phi\right\rangle.
$$
Hence, $({W}_n(t))'=W_n(-t)$.
Therefore, representation~\eqref{8.1} follows from  \eqref{49} and \eqref{xi}, since
$$
\left\langle W_n(t-s)\psi^0_n, \partial_l\rho_n\right\rangle
=\left\langle W_n(t)\psi_n^0,W'_n(-s)\partial_l\rho_n \right\rangle
= \left\langle W_n(t)\psi_n^0,
W_n(s)\partial_l\rho_n \right\rangle.
$$

(ii) Let $\chi\in C_0^\infty(\mathbb{R}^3; \mathbb{C}^4)$ with $\mathrm{supp} \chi\subset B_R$.
By \eqref{2.10}, we have
\begin{equation}\label{6.7}
\langle \psi_n(\cdot\,,t),\chi\rangle=
\langle W_n(t)\psi_n^0,\chi\rangle+
\sum\limits_{k=1}^3 \int\limits_0^{t}
q_k(t-s)\left\langle W_n(s)\, i\,\partial_k\rho_n,\chi \right\rangle ds.
\end{equation}
Condition \textbf{A1} and  \eqref{eqfreeD} imply
\begin{equation}\label{6.8}
|\langle W_n(s) \partial_k\rho_n,\chi\rangle|\le C\,\langle s\rangle^{-3/2}\Vert\nabla\rho_n\Vert_{0,R_\rho},
\end{equation}
where  $C=C(R,R_\rho)<\infty$ is a positive constant.
Hence, using \eqref{8.1} and \eqref{6.8}, we obtain
\begin{align}\label{7.8}
\mathbb{E}\left|\int\limits_0^{t}
\Big(q_k(t-s)-\sum\limits_{r=1}^N
\left\langle W_r(t-s)\psi_r^0,\Xi^0_{rk}\right\rangle\Big)
\left\langle W_n(s)\,i\,\partial_k\rho_n,\chi\right\rangle ds\right|^2\notag\\
\le C\left(\int\limits_0^{t} \sqrt{\mathbb{E}|r_0(t-s)|^2}\,\langle s\rangle^{-3/2}ds\right)^2
\le C_1(1+t)^{-1}.
\end{align}
The next step in proving \eqref{83} is to verify that
\begin{equation}\label{6.9}
\mathbb{E}\Big|\int\limits_t^{+\infty}\sum\limits_{r=1}^N
\langle W_r(t-s)\psi_r^0,\Xi_{rk}^0\rangle
\langle W_n(s)\,i\,\partial_k\rho_n,\chi\rangle\,ds\Big|^2\le C(1+t)^{-1}.
\end{equation}
Indeed, by \eqref{16} and \eqref{2.8}, we have
\begin{align*}
\mathbb{E}|\langle W_r(t)\psi_r^0,f\rangle|^2&=
\mathbb{E}|\langle \psi_r^0,W'_r(t)f\rangle|^2=
\sum\limits_{i,j=1}^8 \left(
Q_{0,rr}^{\nu,ij}(x,y),\mathcal{R}^i(W'_r(t)f(x))\,\mathcal{R}^j (W'_r(t)f(y))
\right)\\
&\le C\Vert W'_r(t)f\Vert^2=C\Vert f\Vert^2,\quad \forall f\in L^2,
\quad r\in \overline{N}.
\end{align*}
 By condition~\textbf{A1}, notion~\eqref{xi} and  bounds~\eqref{decayA} and \eqref{2.8},
 \begin{equation}\label{57'}
 \Xi^j_{rk}\in L^2(\mathbb{R}^3;\mathbb{C}^4).
 \end{equation}
Hence, $
\mathbb{E}|\langle W_r(\tau)\psi_r^0,\Xi_{rk}^0\rangle|^2
 \le C\Vert \Xi_{rk}^0\Vert^2\le C_1 <\infty.
$
Together with \eqref{6.8}, this implies  bound~\eqref{6.9}.
Hence, representation \eqref{6.7} and  bounds \eqref{7.8} and \eqref{6.9} imply
\begin{equation*}
\langle \psi_n(\cdot\,,t),\chi\rangle=
\langle W_n(t)\psi_n^0,\chi\rangle+
\sum\limits_{r=1}^N\sum\limits_{k=1}^3 \int\limits_0^{+\infty}
\left\langle W_r(t-s)\psi_r^0,\Xi_{rk}^0\right\rangle
\left\langle W_n(s)\,i\,\partial_k\rho_n,\chi\right\rangle ds+r(t),
\end{equation*}
where $\mathbb{E}|r(t)|^2\le C(1+t)^{-1}$.
Finally,
$\left\langle W_r(t-s)\psi^0_r, \Xi_{rk}^0\right\rangle
= \left\langle W_r(t)\psi_r^0,
W_r(s)\Xi_{rk}^0 \right\rangle$.
Therefore,  representation~\eqref{83} holds by \eqref{teta}.
\end{proof}
\begin{corollary}\label{c7.2}
Let $Z=(\chi,u,v)\in {\cal D}={\cal D}_0\times \mathbb{R}^3\times \mathbb{R}^3$.
Then
$$
\langle Y(t),Z\rangle=
 \langle \mathbf{W}(t) \psi^0, \chi^Z\rangle + r(t),\quad \mathbb{E}|r(t)|^2\le C (1+t)^{-1},\quad t>0,
$$
where $\langle Y(t),Z\rangle= \langle \psi(\cdot\,,t),\chi\rangle
+ q(t)\cdot u+\dot{q}(t)\cdot v$,
$Y(t)=(\psi(\cdot\,,t),q(t),\dot q(t))$ is a solution to the  problem~\eqref{CP},
the function $\chi^Z$ is defined in \eqref{hn}.
 Note that
 \begin{equation}\label{55}
 \chi^Z\in [L^2(\mathbb{R}^3;\mathbb{C}^4)]^N\quad \mbox{for any }\,\, Z\in \mathcal{D}.
 \end{equation}
 Indeed, $\Xi_{nk}^j\in L^2$. Hence,
 $\Vert W_n(s)\Xi_{nk}^0\Vert=\Vert \Xi_{nk}^0\Vert\le C<\infty$.
Let $\mathrm{supp}\,\chi\subset B_{R_1}$. Then,
notation~\eqref{teta} and bounds~\eqref{2.8} and \eqref{6.8}
imply that $\Vert\theta^\chi_{nr}\Vert\le C\Vert \chi\Vert_{0,R_1}$.
\end{corollary}

\section{Proof of Theorem \ref{tA}}
\subsection{Compactness of the measures $\mu_t$}\label{sec4.1}
Lemma~\ref{l2.1} follows
from the bound~\eqref{7.1.1} by using the Prokhorov Theorem,
 \cite[Lemma II.3.1]{VF}, and technique of the proof in
 \cite[Thm.~XII.5.2]{VF}.
Now we prove the bound~\eqref{7.1.1}.
Let $U_0(t):Y_0\to Y(t)$ be the strongly continuous group of bounded linear operators
on $\mathcal{E}$ corresponding to the case $\rho\equiv0$.
Then, $U_0(t)Y_0=(\psi_0(\cdot,t),q_0(t),\dot{q}_0(t))$,
where
\begin{equation}\label{3.5}
\psi_0(x,t)=(\psi_{01}(x,t),\ldots,\psi_{0N}(x,t))\equiv \mathbf{W}(t)\psi^0,\quad
\psi_{0n}(x,t)=W_n(t)\psi_n^0,
\end{equation}
 the operator $\mathbf{W}(t)$ is defined in \eqref{W(t)},
and $q_0(t)$ is a solution to the Cauchy problem
$$
\ddot q_0(t)+V q_0(t)=0,\quad t\in \mathbb{R},\quad
(q_0(t),\dot q_0(t))|_{t=0}=(q^0,p^0).
$$
Hence,
\begin{equation}\label{3.5'}
q_0(t)=\cos(\sqrt{V}t)q^0+V^{-1/2}\sin(\sqrt{V} t)p^0.
\end{equation}
At first, we prove that
\begin{equation}\label{7.1.10}
\sup\limits_{t\ge 0} \mathbb{E}\Vert U_0(t) Y_0\Vert^2_{{\cal E},R}
\le C(R),\quad\forall R>0.
\end{equation}
Indeed,
$
\Vert U_0(t) Y_0\Vert^2_{{\cal E},R}=
\Vert \mathbf{W}(t) \psi^0\Vert^2_{0,R}+|q_0(t)|^2+|\dot q_0(t)|^2
$.
By condition~{\bf S2},
\begin{equation}\label{57}
\mathbb{E}(|q_0(t)|^2+|\dot q_0(t)|^2)\le C\,\mathbb{E}(|q^0|^2+|p^0|^2)<\infty.
\end{equation}
We verify that
\begin{equation}\label{6.3}
\sup\limits_{t\ge 0} \mathbb{E}\Vert \mathbf{W}(t) \psi^0\Vert^2_{0,R}=
\sup\limits_{t\ge 0} \sum\limits_{n=1}^N
\mathbb{E}\Vert W_n(t) \psi_n^0\Vert^2_{0,R}\le C(R),
\quad\forall R>0.
\end{equation}
To prove \eqref{6.3}, we put $e_t(x):=\mathbb{E}|W_n(t) \psi_n^0(x)|^2$.
Then by condition~\textbf{S3} and relation~\eqref{25}, $e_t(x)=e_t$ for almost every $x\in\mathbb{R}^3$.
Hence, using bound~\eqref{24} and condition~\textbf{S2}, we obtain
$$
e_t|B_R|=\mathbb{E}\Vert W_n(t)\psi_n^0\Vert^2_{0,R}\le \mathbb{E}\Vert \psi_n^0\Vert^2_{0,R+t}
\le e_0|B_{R+t}|,\quad t\ge0,
$$
where $|B_R|$ is the volume of the ball $B_R=\{x\in \mathbb{R}^3:|x|<R\}$.
Hence, $e_t\le e_0|B_{R+t}|/|B_R|$.
As $R\to\infty$, we get $e_t\le e_0$.
Therefore,
$\mathbb{E}\Vert W_n(t) \psi_n^0\Vert^2_{0,R}=e_t|B_R|\le e_0|B_R|\le C(R)<\infty$,
and the bound~\eqref{6.3} is proved.
Further, we represent the solution to problem~\eqref{CP} as
\begin{equation}\label{6.4}
U(t)Y_0=U_0(t) Y_0+\int\limits_0^t
U(t-s) BU_0(s)Y_0\,ds,
\end{equation}
where
\begin{equation}\label{2.6}
B(Y):=\Big(i\, q\cdot\nabla\rho_1,\ldots,i\, q\cdot\nabla\rho_N,0,\
              \sum\limits_{n=1}^N\langle \psi_n,\nabla\rho_n\rangle\Big),
              \quad Y=(\psi_1,\ldots,\psi_N,q,p).
\end{equation}
Hence, \eqref{delocen} and \eqref{7.1.10} yield
\begin{align*}
\mathbb{E}\,\Vert U(t)Y_0\Vert^2_{{\cal E},R}
&\le \mathbb{E}\Vert U_0(t) Y_0\Vert^2_{{\cal E},R}
+\mathbb{E}\int\limits_0^t \Vert U(t-s) BU_0(s)Y_0\Vert^2_{{\cal E},R}\,ds
\nonumber\\
&\le C(R)+\int\limits_0^t \langle t-s\rangle^{-3/2}\,
\mathbb{E}\,\Vert U_0(s) Y_0\Vert^2_{{\cal E},R_\rho}\,ds\le C_1(R)<\infty.
\nonumber
\end{align*}
The bound~~\eqref{7.1.1} is proved. This implies the assertion of Lemma~\ref{l2.1}
because the embedding ${\cal E}\equiv{\cal E}^0\subset {\cal E}^{-\varepsilon}$ is compact for every $\varepsilon>0$.

\subsection{Convergence of correlation functions}\label{s.conv}
To prove \eqref{concorf}, it suffices to prove the convergence of
$\mathbb{E}|\langle Y(t),Z\rangle|^2$ to a limit as $t\to\infty$.
 Corollary~\ref{c7.2} implies that, for any $Z\in {\cal D}$,
\begin{equation}\label{59}
\mathbb{E}|\langle Y(t),Z\rangle|^2=
\mathbb{E}|\langle \mathbf{W}(t)\psi^0,\chi^Z\rangle|^2+o(1)
={\cal Q}^\nu_t(\chi^Z,\chi^Z)+o(1),
\quad t\to\infty,
\end{equation}
where $\chi^Z$ is defined in \eqref{hn} and
${\cal Q}^\nu_t(\chi,\chi):=\mathbb{E}|\langle \mathbf{W}(t)\psi^0,\chi\rangle|^2$.
In \cite[Proposition~4.1]{DKM-03}, we proved the convergence of
${\cal Q}^\nu_t(\chi,\chi)$
to a limit for $\chi\in{\cal D}_0$.
However, generally, $\chi^Z\not\in{\cal D}_0$.
Note that  $\chi^Z\in L^2$ if $Z\in {\cal D}$, due to \eqref{55}.
Now we check that ${\cal Q}^\nu_t(\chi^Z,\chi^Z)$, $t\in \mathbb{R}$,
are equicontinuous in $L^2$.
\begin{lemma}\label{l}
(i) The quadratic form ${\cal Q}^\nu_t(\chi,\chi)=
\int|\langle
\psi^0,\chi\rangle|^2\nu_t(d\psi^0)$, $t\in \mathbb{R}$,
are equicontinuous in $L^2$.
(ii) The characteristic functionals $\widehat{\nu}_t(\chi)$, $t\in \mathbb{R}$, are equicontinuous in $L^2$.
\end{lemma}
\emph{Proof}. (i) It suffices to prove the uniform bound
\begin{equation}\label{6.12}
\sup\limits_{t\in \mathbb{R}} |{\cal Q}^\nu_t(\chi,\chi)|\le C
\Vert\chi\Vert^2,\quad \chi\in L^2.
\end{equation}
We note that
$
 {\cal Q}^\nu_t(\chi,\chi)
=\sum\limits_{n,n'=1}^N
\langle q^\nu_{0,nn'}(x-y), W'_n(t)\chi_n(x)\otimes W'_{n'}(t)\chi_{n'}(y)\rangle.
$
Since ${W}'_n(t)=W_n(-t)$, then using Lemma~\ref{r2.9} and bound~\eqref{2.8} gives
$$
\sup\limits_{t\in \mathbb{R}} |{\cal Q}^\nu_t(\chi,\chi)|
\le C \sup\limits_{t\in \mathbb{R}} \sum\limits_{n=1}^N
\Vert W'_n(t)\chi_n\Vert^2\le C\Vert\chi\Vert^2.
$$

(ii) Applying the Cauchy--Schwartz inequality and bound~\eqref{6.12}, we obtain that
\begin{align*}
|\widehat{\nu}_t(\chi_1)-\widehat{\nu}_t(\chi_2)|&=
\left|\int \left( e^{i\langle \psi^0,\chi_1 \rangle}-
e^{i\langle \psi^0,\chi_2 \rangle}\right)\nu_t(d\psi^0)\right|
\le
\int |\langle \psi^0,\chi_1-\chi_2 \rangle|\,\nu_t(d\psi^0)\\
&\le \sqrt {\int |\langle \psi^0,\chi_1-\chi_2 \rangle|^2\nu_t(d\psi^0)}
=
\sqrt {{\cal Q}^\nu_t(\chi_1-\chi_2, \chi_1-\chi_2)}
\le C\Vert\chi_1-\chi_2 \Vert.
\end{align*}

Thus, formula~\eqref{59} and the item~(i) of Lemma~\ref{l} imply  that
$\lim\limits_{t\to\infty}{\cal Q}^\nu_t(\chi^Z,\chi^Z)
={\cal Q}^\nu_\infty(\chi^Z,\chi^Z)$.
The assertion~(1) of Theorem~\ref{tA} is proved.

\subsection{Convergence of characteristic functionals}\label{sec4.3}
To prove the assertion~(2) of Theorem~\ref{tA}
it remains to prove Lemma~\ref{l2.2}.
 By triangle inequality, we have
\begin{align}
&\left|\mathbb{E}\, e^{i\langle Y(t),Z\rangle}-
\exp\left\{-\frac{1}{2} {\cal Q}_\infty(Z,Z)\right\}\right|
\le\left|\mathbb{E} \left(e^{i\langle Y(t),Z\rangle}-
 e^{i\langle \mathbf{W}(t)\psi^0,\chi^Z\rangle}\right)\right|
\nonumber\\
&+ \left|\mathbb{E}\, e^{i\langle \mathbf{W}(t)\psi^0,\chi^Z\rangle}
-\exp\left\{-\frac{1}{2}{\cal Q}_\infty (Z,Z)\right\}\right|.\label{8.16}
\end{align}
The first term in the right hand side of \eqref{8.16} is estimated by
\begin{align*}
&\Big|\mathbb{E} \Big(e^{i\langle Y(t),Z\rangle}-
e^{i\langle \mathbf{W}(t)\psi^0,\chi^Z\rangle}\Big)\Big|
\le \mathbb{E}\Big|\langle Y(t),Z\rangle
-\langle \mathbf{W}(t)\psi^0,\chi^Z\rangle\Big|
\nonumber\\
&\le \mathbb{E}|r(t)|
\le \Big(\mathbb{E}|r(t)|^2\Big)^{1/2} \le C (1+t)^{-1/2}\to0\quad
\mbox{as }\,\,t\to\infty,
\end{align*}
by Corollary~\ref{c7.2}. It remains to prove the convergence of
$\mathbb{E}\left(\exp\{i\langle \mathbf{W}(t)\psi^0,\chi^Z\rangle\}\right)\equiv\widehat{\nu}_t(\chi^Z)$
to a limit as $t\to\infty$.
Since $\chi^Z\in L^2$, then Theorem~\ref{l6.2}
and the item~(ii) of Lemma~\ref{l} yield
$$
\widehat{\nu}_t(\chi^Z)
\to \exp\left\{-\frac{1}{2}{\cal Q}^\nu_{\infty}(\chi^Z,\chi^Z)\right\}\quad
\mbox{as }\,\, t\to\infty,
$$
where ${\cal Q}^\nu_{\infty}$ is defined by \eqref{qpp}.
This completes the proof of Theorem~\ref{tA}.

\section{Conclusion}
In this paper, we prove the convergence~\eqref{1.8}
assuming that the phase space is $\mathcal{E}=\mathcal{H}\oplus \mathbb{R}^6$
with $\mathcal{H}=[L^2_\mathrm{loc}(\mathbb{R}^3;\mathbb{C}^4)]^N$.
Instead of $\mathcal{E}$, we now introduce a space
$\mathcal{E}_\sigma$.
\begin{definition}
We write $\mathcal{H}_\sigma:=[L^2_\sigma]^N$, $\sigma\in\mathbb{R}$,
where $L^2_\sigma\equiv L^2_\sigma(\mathbb{R}^3;\mathbb{C}^{4})$ are
the weighted Agmon spaces of the complex- and vector-valued functions $\psi$ with the finite norm
$$
\Vert\psi\Vert_\sigma:=\Vert\langle x\rangle^{\sigma}\psi\Vert
=\Big(\int\limits_{\mathbb{R}^3}\langle x\rangle^{2\sigma}|\psi(x)|^2\,dx\Big)^{1/2}<\infty,
\quad \langle x\rangle:=\sqrt{1+|x|^2}.
$$
Below by $L^2_\sigma$ we denote $L^2_\sigma(\mathbb{R}^3;\mathbb{C}^{n})$
with any $n\in \mathbb{N}$.

Let  $\mathcal{E}_\sigma:=\mathcal{H}_\sigma\oplus \mathbb{R}^6$
be  the phase space of $Y=(\psi,q,p)$ with the finite norm
$\Vert Y\Vert_{\mathcal{E}_\sigma}:=\Vert\psi\Vert_\sigma+|q|+|p|$.

For $s,\sigma\in\mathbb{R}$, we denote
$\mathcal{H}^s_\sigma=[H^s_\sigma(\mathbb{R}^3;\mathbb{C}^4)]^N$,
where $H^s_\sigma$ is the weighted Sobolev space with the finite norm
$\Vert\psi\Vert_{H^s_\sigma}=\Vert \langle x\rangle^\sigma\Lambda^s\psi\Vert<\infty$.
Write $\mathcal{E}^s_\sigma=\mathcal{H}_\sigma^s\oplus \mathbb{R}^6$
and $\mathcal{E}_\sigma:=\mathcal{E}^0_\sigma$.
\end{definition}

In this section, we assume  that $\sigma<-3/2$ if  the coupling function $\rho$
satisfies the conditions~\textbf{A1}--\textbf{A3}, and
$\sigma<-5/2$ if instead of \textbf{A1}
we impose a weaker condition~\textbf{A1'}:
\begin{itemize}
\item[{\bf A1'}] $\rho(-x)=\rho(x)$ and $\nabla\rho\in L^2_{-\sigma}$
with some $\sigma<-5/2$.
\end{itemize}

Then, the convergence~\eqref{1.8} holds in the spaces
$\mathcal{E}^{-\varepsilon}_{\bar{\sigma}}$ with any $\varepsilon>0$ and $\bar{\sigma}<\sigma$.
This assertion can be proved by a similar way as Theorem~\ref{tA} with the following two modifications.

At first, we note that the proof of Theorem~\ref{l7.1} and Corollary~\ref{c7.2}
was based on the bound~\eqref{delocen} for the solutions to problem~\eqref{CP}.
If we choose the space $\mathcal{E}_\sigma$ as the phase space,
then we use the following bound for the solutions (see Appendix~A)
\begin{equation}\label{71}
\Vert U(t)Y_0\Vert_{\mathcal{E}_\sigma}\le C\langle t\rangle^{-3/2}\Vert Y_0\Vert_{\mathcal{E}_{-\sigma}},
\end{equation}
where $\sigma<-3/2$ if  conditions~\textbf{A1}--\textbf{A3} hold, and
$\sigma<-5/2$ if conditions~\textbf{A1'}, \textbf{A2}, \textbf{A3} hold.
In particular, instead of the bound~\eqref{6.8}, we apply the following estimate
$$
|\langle W_n(s) \partial_k\rho_n,\chi\rangle|\le
\Vert W_n(s) \partial_k\rho_n\Vert_{L^2_\sigma}\, \Vert \chi\Vert_{L^2_{-\sigma}}
\le
C\,\langle s\rangle^{-3/2}\Vert\nabla\rho_n\Vert_{L^2_{-\sigma}}\quad \mbox{for }\,\,\sigma<-3/2.
$$
Then, the assertions of Theorem~\ref{l7.1} and Corollary~\ref{c7.2} can be proved
by a similar way as in Sec.~\ref{sec3}.

Secondly, to prove Theorem~\ref{tA}, we have to verify the following uniform bound
(instead of  bound~\eqref{7.1.1}):
\begin{equation}\label{72}
\sup\limits_{t\ge 0} \mathbb{E}\,\Vert U(t) Y_0\Vert^2_{\mathcal{E}_\sigma}\le C<\infty.
\end{equation}
In turn, it suffices to prove this bound only for the operator $U_0(t)$
introduced in Sec.~\ref{sec4.1}, and
then to apply  representation~\eqref{6.4} and  bound~\eqref{71}.
Now we derive bound~\eqref{72} for $U_0(t)$.

Since $\Vert U_0(t) Y_0\Vert^2_{\mathcal{E}_\sigma}=
\Vert \mathbf{W}(t) \psi^0\Vert^2_{\mathcal{H}_\sigma}+|q_0(t)|^2+|\dot q_0(t)|^2$,
then \eqref{57} implies that it is enough to prove that
$\mathbb{E}\Vert W_n(t) \psi^0_n\Vert^2_{L^2_\sigma}\le C<\infty$
for any $n\in \overline{N}$.
Write $\psi_{0n}(x,t):=W_n(t) \psi^0_n$.
Using \eqref{ReW} and \eqref{15} gives
$$
\mathbb{E}\left(\widehat{\mathcal{R} \psi_{0n}}(k,t)\otimes\widehat{\mathcal{R} \psi_{0n}}(k',t)\right)
=(2\pi)^3\,\delta(k+k')\, \mathcal{G}_{t,n}(k)\,\widehat{q}^{\nu}_{0,nn}(k)\,\mathcal{G}^*_{t,n}(k),
$$
where we denote
$$
\mathcal{G}_{t,n}(k):=F_{x\to k}\left[(\partial_t-\Lambda_n(\nabla))g_{t,n}(x)\right]=
\cos \omega_n(k) t-\frac{\sin\omega_n(k)t}{\omega_n(k)}\Lambda_n(-ik)
$$
and $\widehat{q}^{\nu}_{0,nn}(k):=(\widehat{q}^{\nu,ij}_{0,nn}(k))_{i,j=1}^8$.
Hence,
$$
q^\nu_{t,n}(x-y):=\mathbb{E}\left(\mathcal{R} \psi_{0n}(x,t)\otimes\mathcal{R} \psi_{0n}(y,t)\right)=
(2\pi)^{-3}\int_{\mathbb{R}^3}  e^{-ik\cdot(x-y)} \,
\mathcal{G}_{t,n}(k)\,\widehat{q}^{\nu}_{0,nn}(k)\,\mathcal{G}^*_{t,n}(k)\,dk.
$$
In particular, by Lemma~\ref{l4.1}, we have
$$
\mathbb{E}|\psi_{0n}(x,t)|^2=\mathrm{tr}[q^\nu_{t,n}(0)]=(2\pi)^{-3}\int_{\mathbb{R}^3}
\mathrm{tr}\left[\mathcal{G}_{t,n}(k)\,\widehat{q}^{\nu}_{0,nn}(k)\,\mathcal{G}^*_{t,n}(k)\right]dk
\le C<\infty.
$$
Therefore,
$$
\mathbb{E}\Vert W_n(t) \psi^0_n\Vert^2_{L^2_\sigma}=
\int_{\mathbb{R}^3} \langle x\rangle^{2\sigma}\,dx\, \mathrm{tr}[q^\nu_{t,n}(0)]\le C<\infty
\quad \mbox{for }\,\,\sigma<-3/2.
$$
By the Prokhorov Theorem and  \cite[Lemma II.3.1]{VF},
 this implies the compactness of $\{\mu_t,\,t\ge0\}$ in the spaces
$\mathcal{E}^{-\varepsilon}_{\bar{\sigma}}$ with
 any $\varepsilon>0$ and $\bar{\sigma}<\sigma$ because
$\mathcal{E}^{0}_{{\sigma}}\subset \mathcal{E}^{-\varepsilon}_{\bar{\sigma}}$
and this embedding is compact.

\section{Appendix~A: Long-time asymptotics}
In this section, we prove the following theorem.
\begin{theorem}\label{tA1}
Let $Y_0\in \mathcal{E}_\sigma$,
$\sigma>3/2$ if conditions~{\bf A1}--{\bf A3} be fulfilled,
and  $\sigma>5/2$ if  conditions~{\bf A1'}, {\bf A2}, {\bf A3} be fulfilled.
Then the solution to problem~\eqref{CP} obeys the following bound:
 \begin{equation}\label{2.9}
\Vert U(t)Y_0\Vert_{\mathcal{E}_{-\sigma}}\le C(1+|t|)^{-3/2}\Vert Y_0\Vert_{\mathcal{E}_{\sigma}},\quad t\in\mathbb{R}.
\end{equation}
\end{theorem}

To prove this theorem,
we apply the  technique of \cite{JK}, which
was developed in the works by Komech \emph{et al}.
\cite{IKV, KSK, KKS, KK}.
This technique is based on the studying the spectral properties of
the resolvent of the stationary problem corresponding to \eqref{CP}.
For details, see Appendix~B.
\begin{proof}
In Appendix~B below, we prove the bound~\eqref{decayA}.
Now we derive bound~\eqref{2.9} using  bound~\eqref{decayA}
and the following  well-known result, which is proved, e.g.,  in \cite[Lemma~15.1]{KKS}.
\begin{lemma}\label{l2.11}
Let $\psi^0_n\in L^2_\sigma$ with $\sigma>3/2$.
Then the solution  to the  problem~\eqref{3}
satisfies the following estimate
\begin{equation}\label{2.7}
\Vert W_n(t)\psi^0_n\Vert_{-\sigma}\le
C\langle t\rangle^{-3/2}\Vert\psi^0_n\Vert_{\sigma},\quad t\ge0.
\end{equation}
\end{lemma}

This bound implies the ``good'' estimates
 for $F(t)$ and $H(t)$ defined in \eqref{41} and \eqref{2.14}.
\begin{corollary}\label{rem2.5}
Let $\nabla\rho\in L^2_\sigma$ with $\sigma>3/2$. Then,
$|H_{kl}(t)|\le C\langle t\rangle^{-3/2}\Vert\nabla\rho\Vert^2_{\sigma}$,
by Lemma~\ref{l2.11}, and
\begin{equation}\label{2.21}
|F_k(t)|\le \sum\limits_{n=1}^N\Vert\partial_k \rho_n\Vert_\sigma\Vert{W}_n(t)\psi_n^0\Vert_{-\sigma}
\le C\langle t\rangle^{-3/2}\Vert\psi^0\Vert_{\sigma},
\quad \mbox{if }\,\, \psi^0=(\psi^0_1,\ldots,\psi_N^0)\in L^2_\sigma.
\end{equation}
If $\psi^0\in L^2$, then bound \eqref{2.8} implies that
$|F_k(t)|\le \sum\limits_{n=1}^N\Vert\partial_k \rho_n\Vert\,\Vert{W}_n(t)\psi_n^0\Vert\le C\Vert\psi^0\Vert$.
\end{corollary}

 If $F(t)\not\equiv0$, then   the representation~\eqref{3.4}, bounds~\eqref{2.21} and \eqref{decayA}
imply the bound~\eqref{2.9} for $q(t)$ and $\dot{q}(t)$.
For the field components $\psi_n(\cdot\,,t)$, the bound~\eqref{2.9}
follows from  representation~\eqref{2.10},  condition~\textbf{A1'}, Lemma~\ref{l2.11}
and bound~\eqref{2.9} for $q(t)$:
\begin{align*}
\Vert\psi_n(\cdot,t)\Vert_{-\sigma}&\le\Vert W_n(t)\psi_n^0\Vert_{-\sigma}
+\int\limits_0^t\Vert W_n(t-s)\nabla\rho_n\Vert_{-\sigma}|q(s)|ds\\
&\le C\langle t\rangle^{-3/2}\Vert \psi_n^0\Vert_{\sigma}
+C_1\int\limits_0^t\langle t-s\rangle^{-3/2}\langle s\rangle^{-3/2}\,ds\,
\Vert \nabla\rho_n\Vert_{\sigma}
\Vert Y_0\Vert_{\mathcal{E}_{\sigma}}
\le C_2\langle t\rangle^{-3/2}\Vert Y_0\Vert_{\mathcal{E}_{\sigma}}.
\end{align*}
Theorem~\ref{tA1} is proved.
\end{proof}

The bound~\eqref{2.9} is useful in the scattering problems for our model.
In particular, using \eqref{2.9}, we prove the following result.
\begin{theorem}\label{tB}
Let the conditions of Theorem~\ref{tA1} hold.
Denote by  $U_0(t)$  the operator $U(t)$ in the case when $\rho\equiv0$.
Then there exist  bounded operators
$\Omega_\pm:{\cal E}_\sigma\to{\cal E}_0$ such that
\begin{equation}\label{1.9}
U(t) Y_0=U_0(t) \Omega_\pm Y_0 +r_\pm(t),\quad t\to\pm\infty,
\end{equation}
where $\Omega_\pm=\lim\limits_{t\to\pm\infty} U_0(-t)U(t)$
and $\Vert r_\pm(t)\Vert_{\mathcal{E}_0}\le C\langle t\rangle^{-1/2}\Vert Y_0\Vert_{\mathcal{E}_\sigma}$.
\end{theorem}
\begin{proof}
To prove  \eqref{1.9},
we apply the classical Cook method (see, e.g., \cite[Sec.~XI.3]{RS3}).
Namely, let $U(t)Y_0=(\psi(\cdot,t),q(t),\dot{q}(t))$ be the solution to problem \eqref{1}--\eqref{2}
with initial data $Y_0=(\psi^0,q^0,p^0)$,
and let $U_0(t):Y_0\to Y(t)$ be the strongly continuous group of bounded linear operators
on $\mathcal{E}$ corresponding to the case $\rho\equiv0$.
Then, $U_0(t)Y_0=(\psi_0(\cdot,t),q_0(t),\dot{q}_0(t))$,
where $\psi_0(x,t)$ is defined in \eqref{3.5} and $q_0(t)$ in \eqref{3.5'}.
Hence, the bound~\eqref{2.8} implies that
\begin{equation}\label{3.6}
\sup_{t\in \mathbb{R}}
\Vert U_0(t)Y_0\Vert_{\mathcal{E}_0}\le C\Vert Y_0\Vert_{\mathcal{E}_0}.
\end{equation}
Furthermore,
the integral Duhamel representation gives
$$
U(t)Y_0=U_0(t)Y_0+\int\limits_0^t U_0(t-s) B( U(s)Y_0)\,ds,\quad t\in\mathbb{R},
$$
where the operator $B$ is defined in \eqref{2.6}.
Note that  $\Vert BY\Vert_{\mathcal{E}_\sigma}\le C\Vert\nabla\rho\Vert_\sigma\Vert Y\Vert_{\mathcal{E}_{-\sigma}}$
by the Cauchy--Schwartz inequality.
Hence, representation~\eqref{1.9} holds with
 $$
 r_\pm(t):= \int\limits_t^{\pm\infty} U_0(t-s) B (U(s)Y_0)\,ds.
 $$
 Indeed,
by formula~\eqref{2.6} and bounds~\eqref{3.6} and \eqref{2.9}, we obtain
 \begin{align*}
 \Vert r_\pm(t)\Vert_{\mathcal{E}_0}&
\le C\int\limits_t^{\pm\infty}\Vert B(U(s)Y_0)\Vert_{\mathcal{E}_0}\,ds
\le C \Vert\nabla\rho\Vert_\sigma
\int\limits_t^{\pm\infty}\Vert U(s)Y_0\Vert_{\mathcal{E}_{-\sigma}}\,ds
\le C_1\int\limits_t^{\pm\infty}
 \langle s\rangle^{-3/2}ds\Vert Y_0\Vert_{\mathcal{E}_\sigma} \\
 &
 \le C_2\langle t\rangle^{-1/2}\Vert Y_0\Vert_{\mathcal{E}_\sigma}.
 \end{align*}
By a similar way, one can check that
 $\Omega_\pm=\lim\limits_{t\to\pm\infty}U_0(-t)U(t)\in{\cal L}({\cal E}_\sigma,{\cal E}_0)$,
 because, formally,
$$
\Omega_\pm Y=Y+\int\limits_0^{\pm\infty} \frac{d}{ds}\Big(U_0(-s)U(s)Y\Big)ds
=Y+\int\limits_0^{\pm\infty} U_0(-s) B( U(s)Y)\,ds,\quad Y\in\mathcal{E}_\sigma.
$$
\end{proof}

Now  we will prove that
\begin{equation}\label{1.7}
U(t)Y_0=\mathcal{P}(W_1(t)\psi^0_1,\ldots,W_N(t)\psi^0_N)+o(1),\quad t\to+\infty,
\end{equation}
where $W_n(t)$, $n\in \overline{N}$, is the solving operator to the Cauchy problem for the free Dirac equation~\eqref{3},
$\mathcal{P}$ is a linear  operator defined in \eqref{Z} below.

 Define a linear operator
$\mathcal{Z}_n:L^2\equiv L^2(\mathbb{R}^3;\mathbb{C}^{4N})\to
L^2_{-\sigma}(\mathbb{R}^3;\mathbb{C}^4)$, $\sigma>3/2$, as
 \begin{equation}\label{Phi}
 \mathcal{Z}_n(\psi):=i\sum\limits_{r=1}^N\sum\limits_{k=1}^3\int\limits_0^{+\infty}
{W}_n(s)\,\partial_k \rho_n(x)\,
 \left\langle \psi_r,{W}_r(s)\,\Xi^0_{rk}\right\rangle\,ds,
\quad n=1,\ldots,N,
\end{equation}
where $\psi=(\psi_1,\ldots,\psi_N)\in L^2$, and $\Xi^0_{rk}$ is defined in \eqref{hn}.
  Hence, by the bounds~\eqref{2.7}, \eqref{57'} and \eqref{2.8}, we have
 \begin{align*}
 \Vert\mathcal{Z}_n(\psi)\Vert_{-\sigma}
& \le\sum\limits_{r=1}^N\sum\limits_{k=1}^3\int\limits_0^{+\infty}
\Vert {W}_n(s)\,\nabla \rho_n\Vert_{-\sigma}
\Vert\psi_r\Vert \,\Vert {W}_r(s)\,\Xi^0_{rk}\Vert\,ds\\
&\le C \sum\limits_{r=1}^N \sum\limits_{k=1}^3\int\limits_0^{+\infty}\langle s\rangle^{-3/2}ds\,
\Vert\nabla\rho_n\Vert_{\sigma}\Vert\psi_r\Vert\,\Vert \Xi^0_{rk}\Vert
 \le C\Vert\psi\Vert.
 \end{align*}
Finally, we introduce  a linear bounded operator
$\mathcal{P}:L^2\to{\cal E}_{-\sigma}$, $\sigma>3/2$, by the rule
\begin{equation}\label{Z}
\mathcal{P}:\psi\to \Big(\psi_1+\mathcal{Z}_1(\psi),\ldots,
\psi_N+\mathcal{Z}_N(\psi),
\sum\limits_{n=1}^N\langle\psi_n,\Xi^0_n\rangle,
\sum\limits_{n=1}^N\langle\psi_n,\Xi^1_n\rangle\Big),
\quad \psi=(\psi_1,\ldots,\psi_N).
\end{equation}
In particular,
$\Vert\mathcal{P}({W}_1(t)\psi^0_1,\ldots,{W}_N(t)\psi^0_N)
\Vert_{-\sigma}\le C\Vert\psi^0\Vert$.
\begin{theorem}\label{tC}
 Let conditions of Theorem~\ref{tA1} hold.
 Then, for $Y_0=(\psi_1^0,\ldots,\psi_N^0, q^0,p^0)\in \mathcal{E}_0$,
 $$
 U(t)Y_0=
 \mathcal{P}\left({W}_1(t)\psi^0_1,\ldots,{W}_N(t)\psi^0_N\right)+r(t),
\quad \mbox{where }\,\,\Vert r(t)\Vert_{\mathcal{E}_{-\sigma}}\le C
 \langle t\rangle^{-1/2}\Vert Y_0\Vert_{\mathcal{E}_0}.
 $$
 If $Y_0\in \mathcal{E}_\sigma$, then
 \begin{equation}\label{8.1}
q^{(j)}_k(t):=\frac{d^j}{dt^j}q_k(t)=\sum\limits_{n=1}^N
\langle {W}_n(t)\psi_n^0,\Xi^j_{nk}\rangle+r_j(t),\quad j=0,1,\quad k=1,2,3,\quad t>0,
\end{equation}
where
$|r_j(t)|\le C\langle t\rangle^{-3/2}\Vert Y_0\Vert_{\mathcal{E}_\sigma}$.
\end{theorem}
\begin{proof}
 At first, we prove the representation~\eqref{8.1},
where
$|r_j(t)|\le C\langle t\rangle^{-\kappa/2}$ with $\kappa=1$
if $Y_0\in{\cal E}_0$, and $\kappa=3$ if $Y_0\in{\cal E}_\sigma$.
Indeed, formula~\eqref{3.4} and bound~\eqref{decayA} imply
\begin{equation}\label{8.4}
  \Big|q^{(j)}(t)-\int\limits_0^t N^{(j)}(s)F(t-s)\,ds\Big|
   \le C\langle t\rangle^{-3/2}(|q^0|+|p^0|).
\end{equation}
Using  bound~\eqref{2.21} if $Y_0\in{\cal E}_\sigma$,
  and the boundedness of the function $F(t)$ if $Y_0\in{\cal E}_0$, we have
 \begin{equation}\label{8.7}
\Big|\int\limits_t^{+\infty}N^{(j)}(s)F(t-s) \,ds\Big|
\le C\left\{
\begin{aligned}
 &\langle t\rangle^{-3/2}\Vert \psi^0\Vert_\sigma,
 &\mbox{if }\,\,\psi^0\in L^2_\sigma,\\
 &\langle t\rangle^{-1/2}\Vert \psi^0\Vert_0,
 &\mbox{if }\,\,\psi^0\in L^2_0,
\end{aligned}
\right.\quad t>0.
\end{equation}
Therefore, the bounds \eqref{8.4} and \eqref{8.7} imply
\begin{equation}\label{8.3}
q^{(j)}_k(t)=\sum\limits_{l=1}^3
\int\limits_0^{+\infty} N^{(j)}_{kl}(s)F_l(t-s)\,ds+r_j(t),\quad t>0,
\quad j=0,1,\quad k=1,2,3,
\end{equation}
where $|r_j(t)|\le C\langle t\rangle^{-1/2}\Vert Y_0\Vert_{\mathcal{E}_0}$ if $Y_0\in \mathcal{E}_0$,
and $|r_j(t)|\le C\langle t\rangle^{-3/2}\Vert Y_0\Vert_{\mathcal{E}_\sigma}$ if $Y_0\in \mathcal{E}_\sigma$.
In turn,
$$
F_l(t-s)=\sum\limits_{n=1}^N\left\langle {W}_n(t-s)\psi_n^0, \partial_l \rho_n\right\rangle=
\sum\limits_{n=1}^N\left\langle {W}_n(t)\psi_n^0, {W}_n(s)\,\partial_l \rho_n\right\rangle,
\quad l=1,2,3,
$$
since $({W}_n(t))'=W_n(-t)$.
Thus, representation~\eqref{8.1} follows from formulas~\eqref{xi} and \eqref{8.3}.

Furthermore, \eqref{2.10} and \eqref{8.1} imply
\begin{equation*}
\psi_n(x,t)={W}_n(t)\psi_n^0
+i\sum\limits_{r=1}^N\sum\limits_{k=1}^3\int\limits_0^{t}
{W}_n(s)\partial_k \rho_n(x)\left(\left\langle {W}_r(t-s)\psi_r^0,\Xi^0_{rk}\right\rangle+r_0(t-s)\right)ds,
\end{equation*}
where   $\Xi^0_{rk}$ is defined in \eqref{xi}.
Let $\psi^0\in L^2$. Then, \eqref{2.7} and \eqref{8.3} imply
\begin{equation*}
\Big\Vert\int\limits_0^{t} {W}_n(s)\partial_k \rho_n \,r_0(t-s)\,ds\Big
\Vert_{-\sigma}
\le C\Vert\nabla\rho_n\Vert_{\sigma}\int\limits_0^{t}\langle s\rangle^{-3/2}
|r_0(t-s)|\,ds\le C\langle t\rangle^{-1/2}\Vert Y_0\Vert_{\mathcal{E}_0}.
\end{equation*}
Furthermore, the bounds \eqref{2.7} and \eqref{2.8} and condition~\textbf{A1'} imply
\begin{align*}
 &\Big\Vert\int\limits_t^\infty {W}_n(s)\partial_k \rho_n
  \left\langle {W}_r(t-s)\psi_r^0,\Xi^0_{rk}\right\rangle \,ds\Big\Vert_{-\sigma}
  \le \int\limits_t^\infty \left\Vert {W}_n(s)\partial_k \rho_n\right\Vert_{-\sigma}
 \Vert {W}_r(t-s)\psi_r^0\Vert_0 \,\Vert\Xi^0_{rk}\Vert_0 \,ds
  \notag\\
  &\le C\int\limits_t^\infty \langle s\rangle^{-3/2}\,ds\,
  \Vert\nabla\rho_n\Vert_{\sigma}
 \Vert\psi_r^0\Vert_0 \,\Vert\Xi^0_{rk}\Vert_0
 \le  C_1\langle t\rangle^{-1/2}\Vert\psi_r^0\Vert_0.
  \end{align*}
Finally, since
$\left\langle {W}_r(t-s)\psi_r^0,\Xi^0_{rk}\right\rangle=
\left\langle {W}_r(t)\psi_r^0,({W}_r(-s))'\,\Xi^0_{rk}\right\rangle
=\left\langle {W}_r(t)\psi_r^0,{W}_r(s)\,\Xi^0_{rk}\right\rangle$, then
$$
i\sum\limits_{r=1}^N\sum\limits_{k=1}^3
\int\limits_0^{+\infty}
{W}_n(s)\partial_k \rho_n(x)\,
\left\langle {W}_r(t-s)\psi_r^0,\Xi_{rk}^0\right\rangle\,ds
=\mathcal{Z}_n\left({W}_1(t)\psi_1^0,\ldots,{W}_N(t)\psi_N^0\right),
$$
where the operator $\mathcal{Z}_n$ is defined in \eqref{Phi}.
Hence,
$$
\psi_n(x,t)={W}_n(t)\psi_n^0
+\mathcal{Z}_n\left({W}_1(t)\psi_1^0,\ldots,{W}_N(t)\psi_N^0\right)
 +r(x,t),\quad t>0,
$$
where $\Vert r(\cdot,t)\Vert_{-\sigma} \le C\langle t\rangle^{-1/2}\Vert Y_0\Vert_{\mathcal{E}_0}$.
\end{proof}

\section{Appendix B: Proof of Theorem~\ref{h-eq}}
In this section, we prove the bound~\eqref{decayA}
assuming that conditions \textbf{A1'}, \textbf{A2} and \textbf{A3} hold.

\begin{lemma}
The  a priori estimate holds,
\begin{equation}\label{3.0}
\Vert Y(t)\Vert_{\mathcal{E}_0}\le Ce^{\gamma |t|}\Vert Y_0\Vert_{\mathcal{E}_0},\quad t\in\mathbb{R},\quad
\gamma:=\tilde{\gamma}\Vert\nabla\rho\Vert,\quad \tilde{\gamma}:=\max\left\{1,\kappa^{-2}\right\}>0,
\end{equation}
where  $\kappa^2$ is a minimal eigenvalue of the matrix $V$.
\end{lemma}
\begin{proof}
Denote
$$
h(t):=\frac12\Big(\sum\limits_{n=1}^N\Vert\psi_n(\cdot,t)\Vert^2+|\dot q(t)|^2+q(t)\cdot Vq(t)\Big),\quad t\in\mathbb{R}.
$$
Then
$\Vert\psi(\cdot,t)\Vert^2+|\dot q(t)|^2+\kappa^2|q(t)|^2\le 2h(t)$.
Hence, $\Vert Y(t)\Vert^2_{\mathcal{E}}\le \widetilde{\gamma} 2h(t)$,
where $\tilde{\gamma}$ is defined in \eqref{3.0}.
Let us assume that
$\psi^0\in C^2_0(\mathbb{R}^3;\mathbb{C}^{4N})$.
Then, by Eqs~\eqref{1} and \eqref{2}, one obtains
$$
\dot  h(t)= \sum\limits_{n=1}^N\langle \psi_n(\cdot,t),i\nabla\rho_n\rangle\cdot q(t)
+\sum\limits_{n=1}^N\langle \psi_n(\cdot,t),\nabla\rho_n\rangle\cdot \dot q(t).
$$
Hence,
$$
|\dot  h(t)|\le \sum\limits_{n=1}^N\Vert\psi_n(\cdot,t)\Vert\,\Vert\nabla\rho_n\Vert\,(|q(t)|+|\dot q(t)|)
\le 2\gamma h(t).
$$
Hence,  the Gronwall inequality implies  the estimate
$h(t)\le e^{2\gamma |t|}h(0)$, and then the bound~\eqref{3.0} holds.
For any $Y_0\in \mathcal{E}_0$, the bound~\eqref{3.0}  follows from the continuity of $U(t)$
and the fact that
$C_0^2(\mathbb{R}^3)\oplus \mathbb{R}^3\oplus \mathbb{R}^3$
is dense in $\mathcal{E}_0$.
\end{proof}

Using \eqref{ReW}, we rewrite $H_{ij}(t)$,
$i,j=1,2,3$, defined in \eqref{2.14} as
\begin{equation*}
H_{ij}(t)=\sum\limits_{n=1}^N
\Big( \partial_i R_n^+,
(\partial_t-\Lambda_n(\nabla))g_{t,n}*\partial_jR_n^-\Big),
\quad \mbox{where }\,\, R_n^+:=\mathcal{R}(\rho_n),\quad R_n^-:=\mathcal{R}(i\rho_n),
\end{equation*}
where $(\cdot\,,\cdot)$ denotes the inner product in $L^2(\mathbb{R}^3)\otimes \mathbb{R}^8$.
To prove Theorem~\ref{h-eq}, we use the technique of
\cite{JK, KSK} and apply the Fourier--Laplace transform to \eqref{CP},
\begin{equation}\label{4.2}
\widetilde{Y}(\lambda)=
\int\limits_{0}^{+\infty} e^{-\lambda t}\,Y(t)\,dt,\quad \Re\lambda>\gamma,
\quad Y(t)=(\psi_1(\cdot\,,t),\ldots,\psi_N(\cdot\,,t),q(t),p(t)),
\end{equation}
with the constant $\gamma>0$ from the bound~\eqref{3.0}.
The integral in \eqref{4.2} converges and is analytic for $\Re\lambda>\gamma$.
We assume that  $\psi_n^0\equiv0$ $\forall n$.
Then, in the Fourier--Laplace transform,
 Eq.~\eqref{cs0} with $F(t)\equiv0$ has a form
\begin{equation}\label{M(lambda)}
\mathcal{N}(\lambda) \tilde{q}(\lambda)=p^0+\lambda q^0,
\qquad
\mathcal{N}(\lambda):=\lambda^2+V-\widetilde{H}(\lambda),
\end{equation}
where $\widetilde{H}(\lambda)$ stands for the $3\times 3$ matrix with the elements
$\widetilde{H}_{ij}(\lambda)$,
\begin{equation}\label{4.0}
\widetilde{H}_{ij}(\lambda)=
\sum\limits_{n=1}^N
\Big( \partial_i R_n^+,
(\lambda-\Lambda_n(\nabla))\widetilde{g}_{\lambda,n}*\partial_jR_n^-\Big),
\qquad i,j=1,2,3.
\end{equation}
Recall that $\widetilde{g}_{\lambda,n}$ stands for
the fundamental solution of the operator $-\Delta+m_n^2+\lambda^2$,
\begin{equation}\label{gn}
\widetilde{g}_{\lambda,n}(x)=\frac{e^{-\kappa_n|x|}}{4\pi|x|},
\quad \kappa_n^2:=m_n^2+\lambda^2,
\quad x\in\mathbb{R}^3,
\end{equation}
where  $\Re\kappa_n>0$ for $\Re\lambda>0$.
Applying \eqref{4.0} and the Fourier transform $\widehat{R}_n^{\pm}(k):=F_{x\to k}[R_n^\pm(x)]$, we have
\begin{equation*}
\widetilde{H}_{ij}(\lambda)=(2\pi)^{-3}\sum\limits_{n=1}^N
\int\limits_{\mathbb{R}^3}\frac{k_i k_j}{k^2+\lambda^2+m_n^2}\,
\widehat{R}_n^+\cdot
\left(\lambda-\Lambda_n(-ik)\right)\widehat{R}_n^-(k)\,dk,
\end{equation*}
where ``$\cdot$'' is the Hermitian product in $\mathbb{C}^8$.
Write $\mu_n=\Re\rho_n$ and $\nu_n=\Im\rho_n$.
By condition~\textbf{{A1'}}, $\widehat{\mu}_n$ and
$\widehat{\nu}_n$ are real valued functions $\forall n$.
Hence,  the explicit formulas for $\widehat{R}_n^\pm$,
$\Lambda_n(-ik)$ and $\alpha_j$ give
\begin{align*}
&\widehat{R}_n^+\cdot
(\lambda-\Lambda_n(-ik))\widehat{R}_n^-(k)
=
m_n(\widehat{\mu_n}\cdot\beta\widehat{\mu_n}+
\widehat{\nu_n}\cdot\beta\widehat{\nu_n})-
\lambda(\widehat{\mu_n}\cdot\widehat{\nu_n}-
\widehat{\nu_n}\cdot\widehat{\mu_n})\notag\\
&-i\sum\limits_{l=1,3}k_l \Big(\widehat{\mu_n}\cdot\alpha_l\,\widehat{\nu_n}-
\widehat{\nu_n}\cdot\alpha_l\,\widehat{\mu_n}\Big)
-k_2\Big(\widehat{\mu_n}\cdot \alpha_2\,\widehat{\mu_n}
+
\widehat{\nu_n}\cdot \alpha_2\,\widehat{\nu_n}\Big)= m_n\,\mathcal{B}_n(k),
\end{align*}
because $\widehat{\rho}_n=\widehat{\mu}_n+i\,\widehat{\nu}_n$ and
$\widehat{\mu_n}\cdot\beta\widehat{\mu_n}+
\widehat{\nu_n}\cdot\beta\widehat{\nu_n}=\widehat{\rho_n}\cdot\beta\widehat{\rho}_n= \mathcal{B}_n(k)$.
Hence,
\begin{equation}\label{4.11}
\widetilde{H}_{ij}(\lambda)=
(2\pi)^{-3}\sum\limits_{n=1}^Nm_n\int\limits_{\mathbb{R}^3}
\frac{k_i k_j\, \mathcal{B}_n(k)}{k^2+\lambda^2+m_n^2}\,dk,
\end{equation}
where $\mathcal{B}_n(k)>0$ by condition~\textbf{A3}.
Moreover, applying the inverse Fourier transform, we obtain 
\begin{equation}\label{4.12}
\widetilde{H}_{ij}(\lambda)=
\sum\limits_{n=1}^N m_n\Big(
\Big(  \partial_i \mu_n,
\widetilde{g}_{\lambda,n}* \beta\,\partial_j \mu_n\Big)+
\Big( \partial_i \nu_n,
\widetilde{g}_{\lambda,n}*\beta\,\partial_j \nu_n\Big)\Big).
\end{equation}
The matrix $\widetilde{H}(\lambda)$ is well defined for $\Re\lambda>0$,
because the denominator in \eqref{4.11} does not vanish.
The following result is proved in \cite[Sec.~13]{IKV}.
\begin{lemma}\label{l7.3}
(i) For $\Re\lambda>0$, the operator $-\Delta+m_n^2+\lambda^2$
is invertible in $L^2(\mathbb{R}^3)$ and its fundamental solution~\eqref{gn}
decays exponentially as $|x|\to\infty$.
(ii)  For every fixed $x\not=0$, the  function $\widetilde{g}_{\lambda,n}(x)$
admits an analytic continuation (in variable $\lambda$) to the Riemann
surface of the algebraic function $\sqrt{\lambda^2+m_n^2}$
with the branching points $\lambda=\pm \,i\,m_n$, $n=1,\ldots,N$.
\end{lemma}

It follows from Lemma \ref{l7.3}, formulas \eqref{M(lambda)} and \eqref{4.11}
that
 $\mathcal{N}(\lambda)$ admits an analytic continuation from the complex half-plane
$\Re\lambda>0$ to the Riemann surface $\Sigma$  with the branching points,
which are projected into the points $\pm\, i\,m_n$, $n\in \overline{N}$.
(Here $\Sigma$ is the $2^{K}$-sheeted surface,
where  $K\in[1,N]$ is the number of pairwise distinct numbers among $m_1,\ldots,m_N$).
 Moreover, the matrix $\mathcal{N}^{-1}(\lambda)$ exists for large $\Re\lambda$,
 since  $\widetilde{H}(\lambda)\to0$
as $\Re\lambda\to\infty$ by \eqref{4.11}.
\begin{corollary}
(i) The matrix $\mathcal{N}(\lambda)$ is invertible for $\Re\lambda>0$, and
\begin{equation}\label{a9}
\widetilde{q}(\lambda)=\mathcal{N}^{-1}(\lambda)
(\lambda q^0+   p^0),\quad \Re\lambda>0.
\end{equation}
(ii) The matrix $\mathcal{N}^{-1}(\lambda)$ admits a meromorphic continuation
from the  half-plane $\Re\lambda>0$ to the Riemann surface $\Sigma$
with the branching points  $\pm\, i\,m_n$, $n=1,\ldots,N$.
\end{corollary}

Formula \eqref{4.12} and the bounds for the convolution operator with the kernels
$\partial^k_\lambda \widetilde{g}_{\lambda,n}$
(see, e.g., \cite[Theorem~8.1]{JK})  imply that
$$
|\partial^k_\lambda\widetilde{H}_{ij}(\lambda)|
\le C\sum\limits_{n=1}^N\Vert\partial_i\rho_n\Vert_{\sigma}
\Vert\partial^k_\lambda\widetilde{g}_{\lambda,n}*\partial_j\rho_n\Vert_{-\sigma}
\le C_k\Vert\nabla\rho\Vert^2_{\sigma}\,|\lambda|^{-(k+1)/2}
\quad \mbox{as }\,\, |\lambda|\to\infty,
$$
 where $\Re\lambda>0$, $k=0,1,2$;
 $\sigma>1$ if $k=0$ and $\sigma>k+1/2$ if $k=1,2$.
(This explains our choice of $\sigma$ as $\sigma>5/2$ if \textbf{A1'} holds).
Together with \eqref{M(lambda)}, this implies the following result.

\begin{lemma} \label{rem4.3}
There is a  matrix-valued function $D(\lambda)$ such that
$\mathcal{N}^{-1}(\lambda)=\lambda^{-2}\mathrm{I}+D(\lambda)$,
where $|\partial_\lambda^k D(\lambda)|\le C_k|\lambda|^{-4}$ for $|\lambda|\to\infty$,
$\Re\lambda>0$,  $k=0,1,2$.
\end{lemma}
\textbf{Remark}.
Let all functions $\rho_n(x)$ be compactly supported (i.e., condition~\textbf{A1} holds).
Then
$$|\partial^k_\omega\widetilde{H}_{ij}(i\omega+0)|
\le C_k\Vert\nabla\rho\Vert^2\,|\omega|^{-1}
\,\,\, \mbox{for }\,\, \omega\in\mathbb{R}:\, |\omega|\ge \max_n m_n+1,
\,\,\,\mbox{and every } k=0,1,2,\ldots,
$$
by \eqref{4.12} and \cite[formula~(16.7)]{IKV}.
Moreover, in this case, all results of Theorems~\ref{tA}--\ref{h-eq}
remain true with $\sigma>3/2$.

Now we investigate the limit values of $\mathcal{N}^{-1}(\lambda)$
at the imaginary axis $\lambda=i\,\omega$, $\omega\in \mathbb{R}$,
applying the methods of \cite[Proposition~15.1]{IKV} and \cite[Lemma~7.2]{KSK}.
Without loss of generality, we assume that $N=2$ and $0<m_1<m_2$.
The other cases can be considered similarly.
The limit matrix
\begin{equation}\label{a3}
\mathcal{N}(i\omega+0)=
-\omega^2\mathrm{I}+ V-\widetilde{H}(i\omega+0),
\quad \omega\in \mathbb{R},
\end{equation}
exists, and its entries are continuous functions of $\omega\in \mathbb{R}$,
smooth for $|\omega|<m_1$, $m_1<|\omega|<m_2$, and $|\omega|>m_2$.
\begin{lemma}\label{l-a}
The limit matrix $\mathcal{N}(i\omega+0)$ is invertible for $\omega\in \mathbb{R}$.
\end{lemma}
\begin{proof}
(i) Let $|\omega|\le m_1$. Then the matrix
$V-\omega^2 \mathrm{I}-\widetilde{H}(i\omega+0)$ is positive definite.
Indeed, for every $v\in \mathbb{R}^3\setminus\{0\}$, we apply the condition {\bf A2}
with $m_*=m_1$ and obtain
\begin{align*}
v\cdot (V-\omega^2 \mathrm{I}-\widetilde{H}(i\omega+0))v
&=v\cdot Vv -\omega^2|v|^2-
 \sum\limits_{n=1}^N \frac{m_n}{(2\pi)^{3}}
\int_{\mathbb{R}^3}\frac{(k\cdot v)^2 \mathcal{B}_n(k)\,dk}{k^2+m_n^2-\omega^2}\\
&\ge v\cdot Vv-m_1^2|v|^2-v\cdot Kv>0.
\end{align*}
Therefore, $\mathcal{N}(i\omega+0) v\not=0$ for all $v\in \mathbb{R}^3\setminus\{0\}$.

(ii) Let $m_1<|\omega|\le m_2$. Then $v\cdot\Im \widetilde{H}(i\omega+0) v\not=0$
for every $v\in\mathbb{R}^3\setminus\{0\}$.
Indeed,
\begin{equation}\label{a6}
\Im \widetilde{H}_{ij}(i\omega+0)= \frac{m_1}{(2\pi)^{3}}\, \Im
\int_{\mathbb{R}^3}\frac{k_ik_j\,\mathcal{B}_1(k)\,dk}{k^2+m_1^2-(\omega-i0)^2}.
\end{equation}
For $\varepsilon>0$, we consider the function
$$
h_{ij}(i\omega+\varepsilon):=
\int_{\mathbb{R}^3}\frac{k_ik_j\,\mathcal{B}_1(k)}{k^2+m_1^2-(\omega-i\varepsilon)^2}\,dk,
\quad |\omega|>m_1.
$$
Denote $D_\varepsilon(k)=k^2+m_1^2-(\omega-i\varepsilon)^2$. For $|\omega|>m_1$,
$D_0(k)=0$ if $|k|=\sqrt{\omega^2-m_1^2}$.
Let us fix a small $\delta>0$ and introduce a cutoff function
$\zeta\in C_0^\infty(\mathbb{R}^3)$
such that $\zeta(k)\ge0$, $\zeta(k)=1$ if $|D_0(k)|<\delta$
and $\zeta(k)=0$ if $|D_0(k)|\ge 2\delta$.
Note that $\Im h_{ij}(i\omega+0)=\Im h^\delta_{ij}(i\omega+0)$, where
$$
h^\delta_{ij}(i\omega+0):=\lim_{\varepsilon\to0}\int_{\mathbb{R}^3}\zeta(k)
\frac{k_ik_j\,\mathcal{B}_1(k)}{D_\varepsilon(k)}\,dk.
$$
Write $a(k)=\sqrt{k^2+m_1^2}$. Assume that $\omega>m_1>0$. Since
$$
\frac1{D_\varepsilon(k)}=\frac1{2a(k)(a(k)-\omega+i\varepsilon)}
+\frac1{2a(k)(a(k)+\omega-i\varepsilon)},
$$
then $\Im h^\delta_{ij}(i\omega+0)=\Im h^\delta_{-}(i\omega+0)$, where
$$
h^\delta_{-}(i\omega+\varepsilon):=\int_{\mathbb{R}^3}\zeta(k)
\frac{k_ik_j\,\mathcal{B}_1(k)}{2a(k)(a(k)-\omega+i\varepsilon)}\,dk.
$$
We rewrite $h^\delta_{-}(i\omega+\varepsilon)$ as
$$
h^\delta_{-}(i\omega+\varepsilon)=
\int\limits_{m_1-\omega}^{+\infty}
\frac{g(y)}{y+i\varepsilon}\,dy,\quad
g(y):=\int\limits_{a(k)-\omega=y}\zeta(k)
\frac{k_ik_j\,\mathcal{B}_1(k)}{2a(k)|\nabla a(k)|}\,dS.
$$
Hence, $\Im h^\delta_{-}(i\omega+0)=-\pi g(0)$
by the Sokhotski--Plemelj formula~\cite[p.~140]{EKS}. Finally, for $\omega>m_1>0$,
$$
\Im h_{ij}(i\omega+0)=\Im h^\delta_{ij}(i\omega+0)
=\Im h^\delta_{-}(i\omega+0)=
-\pi\int\limits_{|k|=\sqrt{\omega^2-m_1^2}}
\frac{k_ik_j\,\mathcal{B}_1(k)}{2|k|}\,dS.
$$
Thus, by \eqref{a6} and condition \textbf{A3}, we obtain that for $\omega\in\mathbb{R}:\,m_1<|\omega|\le m_2$,
\begin{equation}\label{a8}
v\cdot\Im \tilde{H}(i\omega+0)v=-\frac{{\rm sign}(\omega)}{8\pi^2}
\frac{m_1}{2\sqrt{\omega^2-m_1^2}}\int\limits_{|k|=\sqrt{\omega^2-m_1^2}}
(v\cdot k)^2\mathcal{B}_1(k)\,dS\not=0.
\end{equation}

(iii) Let $|\omega|>m_2$. Then,
$$
v\cdot\Im \tilde{H}(i\omega+0)\,v=-\frac{{\rm sign}(\omega)}{16\pi^2}
\sum\limits_{n=1}^2\frac{m_n}{\sqrt{\omega^2-m_n^2}}\int\limits_{|k|=\sqrt{\omega^2-m_n^2}}
(v\cdot k)^2\mathcal{B}_n(k)\,dS\not=0.
$$

In the general case when $N=1,2,\dots,$ we enumerate $m_1,\dots,m_N$
in the increasing order, $0\le m_1\le m_2\le\dots\le m_N$.
Hence, if $m_d\not=m_{d+1}$ and $m_d<|\omega|\le m_{d+1}$ ($d=1,2,\ldots,N-1$), then
\begin{equation}\label{a2}
v\cdot\Im \tilde{H}(i\omega+0)\,v=-\frac{{\rm sign}(\omega)}{16\pi^2}
\sum\limits_{n=1}^d \frac{m_n}{\sqrt{\omega^2-m_n^2}}\int\limits_{|k|=\sqrt{\omega^2-m_n^2}}
(v\cdot k)^2\mathcal{B}_n(k)\,dS\not=0.
\end{equation}
For $|\omega|>m_N$, formula \eqref{a2} holds with $d=N$.
Thus, formula \eqref{a3} and estimate \eqref{a2} imply
Lemma~\ref{l-a}.
\end{proof}

{\bf Remark}.
We note that condition {\bf A3} is used only in the estimates \eqref{a8} and \eqref{a2}.
Hence, instead of condition {\bf A3} it suffices to assume that
for any $n\in \overline{N}$, $v\in\mathbb{R}^3\setminus \{0\}$ and $r>0$,
$$
\int\limits_{|\theta|=1}
(v\cdot \theta)^2 \mathcal{B}_n(r\theta)
\,dS_{\theta}\not=0.
$$

Lemma~\ref{l-a} implies that the matrix $\mathcal{N}^{-1}(i\omega+0)$ is
a smooth function of  $\omega\in \mathbb{R}$
outside the points $\omega=\pm \,m_n$, $n\in \overline{N}$.

\begin{lemma} \label{16.2}
The matrix $\mathcal{N}^{-1}(\lambda)$ admits the  Puiseux expansion
in a neighborhood of the points $\lambda=\pm \,i\,m_n$, $n=1,\ldots, N$:
\begin{equation*}
\left.
\begin{aligned}\mathcal{N}^{-1}(\lambda)&=c_{\pm,n}+ \mathcal{O}(|\lambda\pm i\, m_n|^{1/2})\\
\partial_\lambda^k\left(\mathcal{N}^{-1}(\lambda)\right)
&= \mathcal{O}(|\lambda\pm i\, m_n|^{1/2-k}),\quad k=1,2
\end{aligned}\right|
\quad \lambda\pm i\, m_n\to0.
\end{equation*}
\end{lemma}

This lemma follows from formulas \eqref{M(lambda)} and \eqref{4.12},
because $\widetilde{g}_{\lambda,n}$ have the corresponding Puiseux expansions
by \eqref{gn}. Using the methods of \cite[Sec.~2]{JK},
 this fact can be proved  by a similar way as \cite[Lemma~16.1]{IKV}
and \cite[Lemma~14.2]{KKS}.

To end the proof of the bound~\eqref{decayA},
we apply  the inverse Laplace transform to \eqref{a9}
and use the technique of \cite{IKV, KKS, KK, V74} and Lemma~10.2 from \cite{JK}. Then,
\begin{equation}\label{a7}
 q(t)=
\frac1{2\pi}\int\limits_{-\infty}^{+\infty}
e^{i\omega t}\mathcal{N}^{-1}(i\omega +0)
(\lambda q^0+ p^0)\,d\omega,\quad t\in\mathbb{R}.
\end{equation}
Without loss of generality, we assume that
$0<m_1<m_2<\dots<m_N$.
We split the Fourier integral \eqref{a7} into $N+1$ terms by using
the partition of unity $\zeta_0(\omega)+\dots+\zeta_N(\omega)=1$,
$\omega\in \mathbb{R}$:
$$
 q(t)=
\frac1{2\pi}\int\limits_{-\infty}^{+\infty}
e^{i\omega t}(\zeta_0(\omega)
+\dots+\zeta_N(\omega))\mathcal{N}^{-1}(i\omega+0)
(\lambda q^0+   p^0)\,d\omega=I_0(t)+\dots+I_N(t),
$$
where
 $\zeta_k(\omega)\in C^\infty(\mathbb{R})$ and
$\mathrm{supp}\,\zeta_0\subset\{\omega\in \mathbb{R}:\,|\omega|>m_N+1 \, \mbox{or }\,|\omega|<m_{1}/2\}$,
$m_n\in \mathrm{supp}\,\zeta_n$, $n=1,\ldots,N$,
$m_k\notin  \mathrm{supp}\,\zeta_n$ if $k\ne n$.
Then
  the function $I_0(t)\in C[0,+\infty)$ decays faster than any power of $t$ by Lemma~\ref{rem4.3}, and
 the functions $I_k(t)\in C^\infty(\mathbb{R})$, $k=1,\ldots,N$,
decay like $\langle t\rangle^{-3/2}$ by Lemma~ \ref{16.2}.
Hence, the bound~\eqref{decayA} holds with $j=0$ and $j=1$.
For $j=2$, this bound can be proved by the similar way.


\end{document}